\documentclass[12pt]{article}

\usepackage{amssymb,amsmath,amsthm,enumerate,comment,mathtools}

\usepackage{graphicx, color}
\usepackage{paralist, enumitem}
\usepackage{geometry}
\usepackage{setspace}
\usepackage{algorithm,tikz}

\usepackage{algpseudocode}

\usepackage{etoolbox}

\usepackage{dsfont,bm}

\usepackage[driverfallback=hypertex,pagebackref=true,colorlinks]{hyperref}
\hypersetup{linkcolor=[rgb]{.7,0,0}}
\hypersetup{citecolor=[rgb]{0,.7,0}}
\hypersetup{urlcolor=[rgb]{.7,0,.7}}

\usepackage[titletoc,title]{appendix}

\usepackage{aliascnt} 

\newtheorem{theorem}{Theorem}[section]
\newtheorem*{theorem*}{Theorem}

\newaliascnt{definition}{theorem}
\newtheorem{definition}[definition]{Definition}
\aliascntresetthe{definition}

\newtheorem*{definition*}{Definition}

\newaliascnt{lemma}{theorem}
\newtheorem{lemma}[lemma]{Lemma}
\aliascntresetthe{lemma}

\newtheorem*{lemma*}{Lemma}

\newaliascnt{claim}{theorem}

\aliascntresetthe{claim}

\newtheorem*{claim*}{Claim}

\newaliascnt{fact}{theorem}
\newtheorem{fact}[fact]{Fact}
\aliascntresetthe{fact}

\newtheorem*{fact*}{Fact}

\newaliascnt{observation}{theorem}

\aliascntresetthe{observation}

\newtheorem*{observation*}{Observation}

\newaliascnt{conjecture}{theorem}

\aliascntresetthe{conjecture}

\newtheorem*{conjecture*}{Conjecture}

\newaliascnt{corollary}{theorem}

\aliascntresetthe{corollary}

\newtheorem*{corollary*}{Corollary}

\newaliascnt{remark}{theorem}

\aliascntresetthe{remark}

\newtheorem*{remark*}{Remark}

\newaliascnt{proposition}{theorem}
\newtheorem{proposition}[proposition]{Proposition}
\aliascntresetthe{proposition}

\newtheorem*{proposition*}{Proposition}

\makeatletter
\patchcmd{\ALG@step}{\addtocounter{ALG@line}{1}}{\refstepcounter{ALG@line}}{}{}
\newcommand{\ALG@lineautorefname}{Line}
\makeatother

\usepackage[draft,inline,marginclue,index]{fixme}  

\fxsetup{theme=color,mode=multiuser}
\FXRegisterAuthor{rrs}{arrs}{\colorbox{orange!90!yellow}{\color{white}RRS}}  
\FXRegisterAuthor{linda}{alinda}{\colorbox{blue}{\color{white}Linda}}  

\DeclarePairedDelimiter{\abs}{\lvert}{\rvert}
\DeclarePairedDelimiter{\set}{\{}{\}}
\DeclarePairedDelimiter{\card}{\lvert}{\rvert}

\DeclarePairedDelimiter{\sq}{[}{]}
\DeclarePairedDelimiter{\paren}{\lparen}{\rparen}

\DeclareMathOperator*{\argmax}{arg\,max}

\DeclareMathOperator*{\snd}{2nd}
\def\maxind{\mathsf{tb}}
\def\E{\mathop{{}\mathbb{E}}}
\def\Var{\mathsf{Var}}
\def\D{\mathcal{D}}
\def\V{\mathcal{V}}

\def\reals{\mathbb{R}}
\def\auc{\mathcal{A}}
\def\barpi{\overline{\pi}}
\def\barp{\overline{p}}
\def\region{\mathcal{R}}
\def\pregion{\mathcal{P}}
\def\T{\mathcal{T}}

\def\rev{\mathsf{Rev}}
\def\srev{\mathsf{SRev}}
\def\bvcg{\mathsf{BVCG}}
\def\vcg{\mathsf{VCG}}
\def\pibvcg{\mathsf{PI}\text{\normalfont{-}}\bvcg}
\def\vv{\varphi}
\def\ivv{\tilde{\vv}}
\def\fv{\Phi}
\def\ub{\mathsf{IU}}

\def\nf{\normalfont{\textsc{NF}}}
\def\und{\normalfont{\textsc{Und}}}
\def\srp{\normalfont{\textsc{Srp}}}

\def\nonfavorite{\normalfont{\textsc{Non-Favorite}}}
\def\under{\normalfont{\textsc{Under}}}
\def\surplus{\normalfont{\textsc{Surplus}}}
\def\core{\normalfont{\textsc{Core}}}
\def\tail{\normalfont{\textsc{Tail}}}
\def\single{\normalfont{\textsc{Single}}}
\def\over{\normalfont{\textsc{Over}}}

\def\ron{\mathsf{Ron}}
\def\util{\mathsf{Util}}
\def\U{\mathsf{U}}
\def\caputil{\widehat{\util}}
\def\capU{\hat{\U}}
\def\fee{\mathsf{Fee}}

\def\N{\mathcal{N}}

\def\iu{\normalfont{\text{IU}}}

\def\dec{\mathsf{dec}}

\allowdisplaybreaks

\title{99\% Revenue with Constant Enhanced Competition}
\author{
Linda Cai\thanks{\href{mailto:tcai@princeton.edu}{tcai@princeton.edu}}\\
{\small Princeton University}
\and
Raghuvansh R. Saxena\thanks{\href{mailto:rrsaxena@princeton.edu}{rrsaxena@princeton.edu}}\\
{\small Princeton University}
}

\date{} 
\begin{document}

\maketitle
\thispagestyle{empty}
\addtocounter{page}{-1}


\begin{abstract}

The {\em enhanced competition} paradigm is an attempt at bridging the gap between simple and optimal auctions. In this line of work, given an auction setting with $m$ items and $n$ bidders, the goal is to find the smallest $n' \geq n$ such that selling the items to $n'$ bidders through a simple auction generates (almost) the same revenue as the optimal auction.

Recently, Feldman, Friedler, and Rubinstein [EC, 2018] showed that an arbitrarily large constant fraction of the optimal revenue from selling $m$ items to a single bidder can be obtained via simple auctions with a constant number of bidders. However, their techniques break down even for two bidders, and can only show a bound of $n' = n \cdot O(\log \frac{m}{n})$.

Our main result is that $n' = O(n)$ bidders suffice for all values of $m$ and $n$. That is, we show that, for all $m$ and $n$, an arbitrarily large constant fraction of the optimal revenue from selling $m$ items to $n$ bidders can be obtained via simple auctions with $O(n)$ bidders. Moreover, when the items are regular, we can achieve the same result through auctions that are prior-independent, {\em i.e.}, they do not depend on the distribution from which the bidders' valuations are sampled.

\end{abstract}

\newpage

%

\section{Introduction} \label{sec:intro}

That optimal auctions are not simple and simple auctions are not optimal is the theme of a lot of recent work on designing auctions for multi-item multi-bidder settings. Indeed, it has been well demonstrated that revenue-optimal auctions selling $m$ items to $n$ additive bidders suffer from several undesirable properties, such as the need for randomization, non-monotonicity, computational intractability, {\em etc.} \cite{Thanassoulis04, ManelliV07, Pavlov11, HartN13, DaskalakisDT14, HartR15, DaskalakisDT17}, that make them impractical. On the other hand, the state of the art bounds for simple auctions only show that they obtain a small constant fraction of the optimal revenue \cite{ChawlaHK07, ChawlaHMS10, ChawlaMS10, HartN12, LiY13, BabaioffILW14, BateniDHS15, Yao15, RubinsteinW15, ChawlaMS15, ChawlaM16, CaiDW16, CaiZ17, EdenFFTW17a}. 

The {\em enhanced competition} paradigm is an attempt at bridging the gap between simple and optimal auctions. In this paradigm, given an auction setting with $m$ items and $n$ independent and identically distributed additive bidders, the goal is to find the smallest number $n' \geq n$ of bidders such that simple auctions with $n'$ bidders (almost) match the revenue of the optimal auction with $n$ bidders. If such a result can be shown, then it conveys the message that an auctioneer aspiring to get the optimal revenue with $n$ bidders need not spend all his energy on finding the optimal, even if impractical, auction with $n$ bidders. Instead, he can try to rope in $n' - n$ more bidders and get the same revenue using a simple and practical auction format.

The focus of this paper is to show enhanced competition results for general auction settings where $n' = O(n)$ is at most a constant times $n$. The first such result is found in the seminal work of Bulow and Klemperer \cite{BulowK96} where it was shown that the revenue of the optimal auction selling a {\em single} item to $n$ bidders is at most the revenue of the simple VCG auction with $n+1$ bidders, as long as the distribution of the bidders' valuation for the item is regular\footnote{A (continuous) distribution with probability density function $f(\cdot)$ and cumulative density function $F(\cdot)$ is regular if the function $x - \frac{ 1 - F(x) }{ f(x) }$ is monotone non-decreasing.}. 

The only other enhanced competition result with $n' = O(n)$ is in \cite{FeldmanFR18} where it is shown that, for any $\epsilon > 0$, a $(1 - \epsilon)$-fraction of the revenue of the optimal auction selling $m$ items to a {\em single} bidder can be obtained by either selling the items separately, or by selling the grand bundle, to a constant number of bidders. However, the techniques used in \cite{FeldmanFR18} do not generalize to $n > 1$ bidders and finding enhanced competition results with $n' = O(n)$ for $n > 1$ bidders remains an open problem\footnote{\cite{FeldmanFR18} also show that, for any $\epsilon > 0$, when $n \gg m$, {\em i.e.}, when the number of bidders is much larger than the number of items, then, even without ``enhancing'' competition, {\em i.e.}, with $n' = n$, the revenue of selling the items separately obtains a $(1 - \epsilon)$-fraction of the revenue of the optimal auction. In other words, these settings are competitive enough for enhanced competition to not yield great gains in the revenue. Thus, the interesting range of parameters is $1 < n \ll m$. See \cite{BeyhaghiW19} for a related result.}.

\subsection{Our Results}
\label{sec:intro:results}

Our main theorem is the first enhanced competition result with $n' = O(n)$ that works for all $m$ and $n$.

\begin{theorem}[Informal]
\label{thm:main}
Consider an auction setting where $m$ items are being sold to $n$ bidders. Let $\epsilon > 0$ and $n' = O(n / \epsilon)$. At least one of the following hold:
\begin{enumerate}
\item \label{thm:main:0.99}A $(1 - \epsilon)$-fraction of the optimal revenue with $n$ bidders is obtained by a VCG auction with $n'$ bidders.
\item \label{thm:main:1} A simple auction (either selling the items separately using Myerson's optimal auction or a VCG auction with an entry fee, as in \cite{CaiDW16}) with $n'$ bidders generates more revenue than the optimal auction with $n$ bidders.
\end{enumerate}
\end{theorem}

Note that aside from Case~\ref{thm:main:0.99} where the optimal revenue with $n'$ bidders is nearly matched by the revenue of a VCG auction with $n$ bidders, \autoref{thm:main} actually promises that simple auctions with $n'$ bidders outperform the optimal auction with $n$ bidders. This is interesting as all known ``hard'' instances for enhanced competition results involve the equal-revenue distribution\footnote{The equal revenue distribution is the distribution defined by $F(x) = 1 - \frac{1}{x}$ for all $x \geq 1$.} for which Case~\ref{thm:main:0.99} does not hold \cite{EdenFFTW17b, FeldmanFR18}. Thus, we outperform the optimal auction in all of these ``hard'' cases.

In fact, when Case~\ref{thm:main:0.99} does not hold, then, at the cost of increasing the total number of bidders by another constant factor, our techniques can also show that a simple auction as in Case~\ref{thm:main:1} with this increased number of bidders obtains much more, say a $100$ times more, revenue than the revenue of the optimal auction with $n$ bidders. As far as we know, this is the first enhanced competition result that not only outperforms the optimal revenue but actually obtains revenue that is significantly larger. At least this stronger version of Case~\ref{thm:main:1} cannot be shown for general distributions, and some condition like Case~\ref{thm:main:0.99} not being true is necessary\footnote{To see why, consider an item such that the values for this item are sampled from a distribution that is supported on the interval $[1,2]$. Even with one bidder, setting a posted price of $1$ achieves revenue $1$, while no number of additional bidders can get revenue larger than $2$.}.

Moreover, this stronger version of \autoref{thm:main} implies that Case~\ref{thm:main:1} can be made to work with any auction that guarantees a constant approximation to the optimal revenue (without enhanced competition). Indeed, the fact that Case~\ref{thm:main:1} obtains $C$ times more revenue than the optimal auction with $n$ bidders (for some $C > 1$) implies that a $C$-approximate auction with $n'$ bidders obtains at least as much revenue as the optimal auction with $n$ bidders\footnote{Thus, we actually show that the two problems: \begin{inparaenum}[(1)] \item \label{item:generalization:1} designing (simple) auctions that guarantee a constant approximation of the optimal revenue, and
\item \label{item:generalization:2} designing (simple) auctions for enhanced competition results 
\end{inparaenum} are equivalent. Indeed, that  (\ref{item:generalization:1}) is necessary for (\ref{item:generalization:2}) is folklore, and (\ref{item:generalization:1}) is sufficient for (\ref{item:generalization:2}) due to the argument above. We could have alternatively presented this equivalence as our main result.}.

Finally, as we discuss in \autoref{sec:overview} below, our proof of \autoref{thm:main} is the same for all values of $m$ and $n$, avoiding the case analysis in \cite{FeldmanFR18} that uses different techniques to prove a claim similar to \autoref{thm:main} in the case $n = 1$ and in the case $n \sim m$.

\paragraph{Prior-independent auctions.} Finally, we mention that even though \autoref{thm:main} matches the optimal revenue with $n$ bidders using a simple auction with $n'$ bidders, the simple auction that it uses in Case~\ref{thm:main:1} is {\em prior-dependent}, {\em i.e.}, the auctioneer needs to know the distribution from which the bidders' sample their valuations in order to run the auction. Dependence on the prior is necessary\footnote{The example to keep in mind is a single-item single-bidder setting where the value of the bidder for the item is sampled from a distribution that, for some $p > 1$, takes the value $0$ with probability $1 - \frac{1}{p}$, and the value $p$ with probability $\frac{1}{p}$. For such a distribution, it is impossible to design an auction with any non-trivial revenue guarantee without the knowledge of $p$.} when there is no other promise on the distribution of the bidders' valuation. However, when these distributions are promised to be regular, there is a long line of work that focuses on developing prior-independent auctions \cite{BulowK96, DevanurHKN11, AzarDMW13,  AzarKW14, GoldnerK16, EdenFFTW17b}. We contribute to this line of work by showing the following prior-independent analogue of \autoref{thm:main}. 

\begin{theorem}[Informal]
\label{thm:mainpi}
Consider an auction setting where $m$ {\em regular} items are being sold to $n$ bidders. Let $\epsilon > 0$ and $n' = O(n / \epsilon)$. At least one of the following hold:
\begin{enumerate}
\item \label{thm:mainpi:0.99}A $(1 - \epsilon)$-fraction of the optimal revenue with $n$ bidders is obtained by a VCG auction with $n'$ bidders.
\item \label{thm:mainpi:1} A prior-independent VCG auction with an entry fee (inspired by \cite{GoldnerK16}) and $n'$ bidders generates more revenue than the optimal auction with $n$ bidders.
\end{enumerate}
\end{theorem}

It is important to mention here that in general, the better of two prior-independent mechanisms is not necessarily prior-independent. However, the stronger version of \autoref{thm:main} also applies to \autoref{thm:mainpi} (along with the other properties mentioned after \autoref{thm:main}) and implies that a prior-independent auction actually follows from \autoref{thm:mainpi}. This is because, when the constant $C$ above equals $\frac{1}{\epsilon}$, one can run the auction in Case~\ref{thm:mainpi:0.99} with probability $1 - \epsilon$ and the auction in Case~\ref{thm:mainpi:1} with probability $\epsilon$ and be guaranteed a $(1 - 2\epsilon)$-fraction of the optimal revenue with $n$ bidders in either case.

Lastly, we can even extend \autoref{thm:mainpi} to certain settings of irregular distributions. Specifically, it is possible to combine our proof of \autoref{thm:mainpi} with ideas from \cite{SivanS13} to get analogous claims for the form of irregular distributions considered there.

\subsection{Related Work}
\label{sec:intro:related}

Besides the works mentioned above, our work is also related to the following works.

\paragraph{Enhanced competition results with $n' = \omega(n)$.} The focus of the current paper is getting enhanced competition results with $n' = O(n)$. However, there is a long line of work focusing on getting enhanced competition results with larger $n'$. Among the first such works were those of \cite{RoughgardenTY20} and \cite{EdenFFTW17b} which show the bounds $n' = m$ and $n' = O(n+m)$ for unit-demand and additive bidders respectively. These works were followed by \cite{FeldmanFR18} and \cite{BeyhaghiW19} which improve these bounds to $n' = n \cdot O(\log \frac{m}{n})$ for additive bidders. We remark that \cite{FeldmanFR18} focuses on getting a $(1 - \epsilon)$-fraction of the optimal revenue with $n$ bidders while all the other works outperform the optimal revenue with $n$ bidders.

One key difference between the current work and the foregoing works is that all of them focus on upper bounding the optimal revenue with $n$ bidders by the revenue of an auction that sells the items separately with $n'$ bidders, while we also consider VCG auctions with an entry fee. It is known that when restricting attention to auctions that sell the items separately, one cannot get a bound better than $n' = n \cdot \Omega(\log \frac{m}{n})$ \cite{FeldmanFR18, BeyhaghiW19}. Thus, these works cannot hope to get $n' = O(n)$ like we do.

\paragraph{The duality framework.} We prove \autoref{thm:main} and \autoref{thm:mainpi} using the duality framework of \cite{CaiDW16}. In this work, \cite{CaiDW16} view the problem of finding the optimal revenue as a linear program, and analyze it in terms of its Lagrangian dual. The duality framework shown in this work is extremely general, and in particular, is the first one that also applies to multi-bidder settings. In fact, it also applies to settings beyond the additive bidder setting we consider in this paper but we shall not need those ideas. 

The \cite{CaiDW16} framework has been used in a lot of subsequent work on getting enhanced competition results in particular and approximately revenue optimal mechanisms in general. Examples include \cite{LiuP18, CaiZ17, EdenFFTW17b, EdenFFTW17a, BrustleCWZ17, DevanurW17, FuLLT18, BeyhaghiW19}. Our duality based proof also uses tools and ideas from \cite{Ronen01, GoldnerK16}.

\subsection{Our Techniques}
\label{sec:intro:tech}

We now summarize the most important ideas in this work, focusing solely on \autoref{thm:main}. A more comprehensive overview can be found immediately below in \autoref{sec:overview}.

As mentioned in \autoref{sec:intro:related} above, several works have studied how many bidders are necessary for the revenue obtained by selling the items separately to surpass the optimal revenue from selling to $n$ bidders. Due to these works, we now know that the answer is $n' = n \cdot \Theta(\log\frac{m}{n})$, and thus any result that works for $n' = O(n)$ (when $n \ll m$) must use auctions other than just selling the items separately. 

The only such enhanced competition result in the literature is by \cite{FeldmanFR18} and shows \autoref{thm:main} when $n = 1$. The proof proceeds by first bounding the optimal revenue when selling to a single bidder by the core-tail decomposition of \cite{LiY13}. The resulting bound is the sum of the welfare from the bidders with a ``low'' value for the items (the \core{}) and the revenue from the bidders with a ``high'' value for the items (the \tail{}). The next step is to upper bound the sum of \core{} and \tail{} by the revenue of a simple auction with $n' = O(1)$ bidders.

The term \tail{} turns out to be easier to bound than \core{}, and the reason is that \core{} corresponds to the {\em welfare} of some distribution and bounding it in terms of the {\em revenue} of some other distribution is like comparing apples to oranges. We get around this problem by adopting a different approach that tries to bound the welfare term coming from the \core{} with another welfare term. Specifically, our main lemma, that works for any $n$ and any distribution, shows that if Case~\ref{thm:main:0.99} of \autoref{thm:main} does not hold, then, the welfare with $n$ bidders can be upper bounded by, say, a $\frac{1}{20}$ fraction of the welfare with $n' = O(n)$ bidders.

Unfortunately, even though \core{} corresponds to the welfare of some distribution and our main lemma applies to welfare of any distribution, it cannot be applied to \core{} as it assumes that the values of all the bidders are drawn independently and identically from the same distribution, a property that \core{} does not satisfy in general. This turns out to be a major obstacle and to get around it, we have to start from scratch. This time, instead of starting with the core-tail decomposition of \cite{LiY13}, we start with the virtual welfare based upper bound on the revenue from \cite{CaiDW16}. We are able to extend our main lemma to this notion of virtual welfare, and show that unless Case~\ref{thm:main:0.99} of \autoref{thm:main} holds, the virtual welfare with $n$ bidders can be upper bounded by a $\frac{1}{20}$ fraction of the virtual welfare with $n' = O(n)$ bidders.

To finish the proof, we use techniques from \cite{CaiDW16} to upper bound the virtual welfare with $n'$ bidders by $20$ times the revenue of simple auctions, {\em i.e.} either selling the items separately or through a VCG auction with an entry fee, with $n'$ bidders.

\subsection{Acknowledgements} The authors are grateful to Matt Weinberg for suggesting the problem and for helpful discussions at multiple stages of this work. We also thank him for his comments on an earlier draft of this manuscript. Last but not the least, we thank the anonymous reviewers for their helpful comments that greatly improved the paper.


\section{Technical Overview} 
\label{sec:overview}

In this section, we cover the main ideas behind the proof of \autoref{thm:main} and \autoref{thm:mainpi}. As the proofs have significant overlap, we focus only on \autoref{thm:main} for the most part.

\subsection{The \cite{FeldmanFR18} Result}
\label{sec:ffr}

As mentioned above, the work of \cite{FeldmanFR18}, showed, amongst other results, that \autoref{thm:main} holds in the case $n = 1$, {\em i.e.}, when there is only one bidder. In this special case, \autoref{thm:main} reduces to showing that 99\% of the revenue can be obtained by either selling the items separately, or selling the grand bundle, to a (large enough) constant number of bidders. The main tool in this result of \cite{FeldmanFR18} is the core-tail decomposition of \cite{LiY13}.

\paragraph{Core-tail decomposition.} The key insight in the core-tail decomposition framework is to set a cutoff $t_j$ for each item $j \in [m]$ and reason separately about the case where the values of the bidders for item $j$ are ``low'', {\em i.e.} at most $t_j$, and when they are ``high'', {\em i.e.}, more than $t_j$. We say that an item $j$ is in the \core{} if its value is low, and say that it is in the \tail{} otherwise. Core-tail decomposition says that, for any choice of the cutoffs $t_j$, the optimal revenue $\rev$ is at most the optimal {\em welfare} from the items in the \core{} and the optimal {\em revenue} from the items in the \tail{}. 

Thus, in order to show that 99\% of the revenue is upper bounded by either the revenue obtained by selling separately, or that obtained by selling the grand bundle to a constant number of bidders, it is sufficient to show that, for some choice of the cutoff, the sum total of the welfare from the \core{} and the revenue from the \tail{} is also upper bounded by the maximum of selling separately or selling the grand bundle to a constant number of bidders. The specific cutoff in \cite{FeldmanFR18} is involved, and for simplicity in this overview, we shall assume that the cutoff is just a large, say the $(1- \delta)$-th for some small $\delta > 0$, quantile.

Bounding the revenue from the \tail{} is the easy part, and uses the observation that the probability that a given item is in the \tail{} is at most $\delta$. In fact, one can show that the revenue from the \tail{} is only a small fraction, say $\sqrt{\delta}$, of the revenue obtained by selling the items separately to a constant number of bidders. 

\paragraph{Bounding the {\core{}}.} The hard part is to bound the welfare from the \core{}, and it is here that the $n=1$ assumption is crucially used. When $n = 1$, the welfare is simply the sum of the bidder's value for all the items in the \core{}. Being the sum of independent random values that are bounded (below the cutoff), one would expect that the welfare from the \core{} would be reasonably well concentrated in an interval of size roughly equal to the standard deviation.

It turns out that, with the right choice of the cutoffs, the standard deviation is much smaller than the revenue obtained by selling the items separately to a constant number of bidders. This implies that either the expected welfare of the \core{} is at most the revenue obtained by selling the items separately to a constant number of bidders, or one can expect (with at least some small constant probability) that the sum of values of a bidder for all the items is very close to the welfare from the \core{}.

In the latter case, when there are many bidders, it is very likely that there {\em exists} a bidder whose total value for all the items is around the welfare from the \core{}. Thus, an auction that sells the grand bundle around this price will likely generate revenue equal to the welfare from the \core{}, proving the result.

\subsection{Difficulty in Extending to Multiple Bidders}
\label{sec:difficulty}

Interestingly, the core-tail decomposition framework used in \cite{FeldmanFR18} extends, with some changes (see \autoref{sec:overviewmultiplebidder}), to the $n > 1$ case. Moreover, the analysis of the \tail{} can be done in essentially the same way, and the only part of the argument above that does not extend to the $n > 1$ case is bounding the \core{}. 

Specifically, what breaks is that the welfare of the \core{} is no longer the sum over all items of the value the (only) bidder has for the item. Instead, the welfare in the $n > 1$ case equals the sum over all items of the {\em maximum} value that any bidder has for the item. Thus, even if one can somehow show that it is highly concentrated around some value, say $x$, it is not clear how to design an auction with a constant-factor more bidders whose revenue is at least $x$. 

In particular, selling the grand bundle at price $x$ does not work, as it is no longer guaranteed that there will be a bidder whose value for the grand bundle will be around $x$. All the concentration of the welfare buys us in this case is that if we take the sum of the maximum value of $n$ bidders for all the items, we are likely to get $x$. However, it is not clear how to realize this maximum value as the revenue of a simple auction.

\subsection{Our Approach}
\label{sec:approach}

 The difficulty described above is serious, as we are trying to upper bound the welfare by the revenue, when it is known that the former may be arbitrarily larger than the latter in the worst case. What is supposed to save us is that the distribution of the values in the \core{} is not a general distribution as, for example, its support is upper bounded by the cutoff. Using this upper bound on the support, prior work \cite[amongst others]{Yao15, CaiDW16} has shown that the welfare of the \core{} is at most a constant (actually, $< 20$) times the revenue of a (simple) auction. However, getting this constant down all the way to $1$, even with more bidders seems challenging.

Our way around this difficulty is to upper bound the welfare of the \core{} in two steps: \begin{inparaenum}[(1)]
\item \label{item:corebound1} Firstly, we show that if ($99\%$ of the) \core{} is larger than revenue obtained by selling separately to a constant-factor larger number of bidders, then the welfare of the \core{} can be bounded by a small constant (say, $\frac{1}{20}$) of the welfare of the \core{} with a constant-factor larger number of bidders. The hope is that proving such a bound is easier as it requires bounding the welfare term by another welfare term.
\item \label{item:corebound2} Next, we invoke the known results about the welfare of the \core{} being at most $20$ times the revenue of simple auctions on the \core{} with the larger number of bidders to get our result about the \core{} with the smaller number of bidders. 
\end{inparaenum}

Overall, this scheme will give us that if ($99\%$ of the) \core{} is larger than the revenue obtained by selling separately to a constant-factor larger number of bidders, then the welfare of the \core{} can be bounded by the revenue of a simple auction with constant-factor larger bidders, as desired.

\subsection{Our Proof When $n = 1$}
\label{sec:overviewsinglebidder}
Step~\ref{item:corebound2} of our two-step approach above follows from known results, and in this subsection, we show Step~\ref{item:corebound1}. We note that our proof for this part does not even require the welfare to come from the \core{} and it actually works for any distribution. In fact, it does not even require $n = 1$ and works for all $n$, but as we shall explain in \autoref{sec:overviewmultiplebidder}, it will only fit in our overall framework when $n=1$.

 The main idea, for the one item case, is captured in the following informal lemma. Using linearity of expectation and the fact that the bidders are additive, a similar lemma can also be shown for the multi-item case.
\begin{lemma}[Informal]
\label{lemma:boostcore} 
Let $\epsilon > 0$. Consider a single item and $n$ bidders each of whose value for the item are sampled from a distribution $D$. If $(1- \epsilon)$ times the welfare of the $n$ bidders is larger than the revenue generated from a second price auction with $\frac{20n}{\epsilon}$ bidders, then, the welfare with $n$ bidders is at most $\frac{1}{20}$ times the welfare with $\frac{20n}{\epsilon}$ bidders.
\end{lemma}
\begin{proof}
Define $n' = \frac{20n}{\epsilon}$ for convenience, and let $v_i$ for $i \in [n']$ be the value of the $i^{\text{th}}$ bidder for the item. As we assume our bidders to be independent and identically distributed, this is just an independent sample from $D$. 

Next, note that as welfare is just the maximum value of the item, we have that the expected welfare with $n$ bidders is just the expected value of $\max_{i \in [n]} v_i$ while the expected welfare with $n'$ bidders is just the expected value of $\max_{i \in [n']} v_i$. Now if the maximizer (we disregard all issues about tie-breaking in this informal lemma and assume that the maximizer is unique) on $v_i$ over all $n'$ bidders lies in the first $n$ bidders, an event (which we denote by $E$) that happens with probability $\frac{n}{n'} = \frac{\epsilon}{20}$ as the bidders are identically distributed, then the maximum value amongst the first $n$ bidders equals that amongst all the $n'$ bidders. Otherwise, the maximum value amongst the first $n$ is at most the second highest value amongst all the $n'$ bidders. We get that:
\[
\E\sq*{ \text{welfare with $n$} } \leq \frac{\epsilon}{20} \cdot \E\sq*{ \text{welfare with $n'$} \mid E } + \Pr\paren*{ \overline{E} } \cdot \E\sq*{ \text{$\snd$ highest value with $n'$} \mid \overline{E} } .
\]
Now, using basic conditional probability, we can upper bound the second term by the expected second highest value overall, which is just the revenue of a second price auction with $n'$ bidders. We get:
\[
\E\sq*{ \text{welfare with $n$} } \leq \frac{\epsilon}{20} \cdot \E\sq*{ \text{welfare with $n'$} \mid E } + \text{Revenue from $\snd$ price auction} .
\]
As we assume that $(1- \epsilon)$ times the welfare of the $n$ bidders is larger than the revenue generated from a second price auction, we get:
\[
\epsilon \cdot \E\sq*{ \text{welfare with $n$} } \leq \frac{\epsilon}{20} \cdot \E\sq*{ \text{welfare with $n'$} \mid E } .
\]
To finish, we remove the conditioning on $E$ noting that the expected value of the maximum is independent of where the maximizer is. This gives:
\[
\E\sq*{ \text{welfare with $n$} } \leq \frac{1}{20} \cdot \E\sq*{ \text{welfare with $n'$} } ,
\]
as claimed in the lemma.
\end{proof}

\subsection{The $n > 1$ Case}
\label{sec:overviewmultiplebidder}

As mentioned above, \autoref{lemma:boostcore} is very general and works for all distributions and all $n$. However, it does require that: \begin{inparaenum}[(i)]
\item \label{item:coreprop1} The distribution for different bidders are independent and identical, so that the probability that $E$ happens is $\frac{n}{n'}$.
\item \label{item:coreprop2} The distribution does not depend on the number $n'$ of bidders participating in the auction, so that the distribution with $n$ bidders is the same as the distribution of the first $n$ bidders when there are actually $n'$ bidders in total.
\end{inparaenum} 

Both these properties hold for the $n = 1$ case, when the distribution of the \core{} for all the bidders is determined by the same cutoff that is not a function of the total number of bidders in the auction. However, core-tail type arguments for the multiple bidder ($n > 1$) case often have a more involved cutoff, that may be different for each bidder, and may also depend on the second highest value (amongst all the bidders) of the item concerned \cite{Yao15, CaiDW16}. As the distribution of the second highest value depends on the number of bidders in the auction, the resulting distribution of the \core{} does not satisfy either of the two properties above. 

Due to these complications in the core-tail framework for the $n > 1$ case, we are forced to adopt a different approach that avoids the core-tail decomposition framework altogether! This becomes possible only because our proof of \autoref{lemma:boostcore} works for all distributions, and not just those that correspond to the core of some other distribution. The goal now is to apply \autoref{lemma:boostcore} on a carefully chosen notion of `virtual welfare'. This notion of virtual welfare must satisfy the following properties in addition to Properties~\ref{item:coreprop1} and~\ref{item:coreprop2} above: \begin{inparaenum}[(i)]\setcounter{enumi}{2}
\item \label{item:coreprop3} It must be an upper bound on the optimal revenue, so that upper bounding it using \autoref{lemma:boostcore} also gives a bound on the optimal revenue.
\item \label{item:coreprop4} It must be within a constant factor of the revenue of a (simple) auction so that, after applying \autoref{lemma:boostcore}, we can upper bound the virtual welfare with $n'$ bidders by the revenue of the simple auction. 
\item \label{item:coreprop5} The virtual values must be at most the corresponding values, so that the expected second highest virtual value is at most the revenue of a second price auction. 
\end{inparaenum}

\paragraph{Duality based virtual values.} To construct such a virtual value function, we use the duality framework of \cite{CaiDW16}. In fact, the work \cite{CaiDW16} itself defines a virtual value function that satisfies Properties~\ref{item:coreprop3},~\ref{item:coreprop4}, and~\ref{item:coreprop5} above but does not satisfy Properties~\ref{item:coreprop1} and~\ref{item:coreprop2}. As a result, we cannot use the results shown in \cite{CaiDW16} as a black box, but are able to suitably adapt them for our purposes.

More specifically, we define a new duality-based revenue benchmark that we call the {\em independent utilities}, or the $\iu$-benchmark. This (randomized) benchmark is parameterized by an integer $k$\footnote{In the actual proof, we set $k = n'$.} and is defined as follows: Consider a bidder $i \in [n]$ and suppose that bidder $i$ has value $v_{i,j}$ for item $j \in [m]$. For each such bidder, we sample valuations for $k-1$ ``ghost''-bidders and simulate a second price auction with bidder $i$ and the ghost bidders. If item $j$ gets him the highest (breaking ties lexicographically) non-negative utility in this auction, the virtual value of bidder $i$ for item $j$ is the Myerson's (ironed) virtual value corresponding to $v_{i,j}$. Otherwise, it is equal to $v_{i,j}$. 

As the ghost bidders are sampled independently and identically for each bidder (and also independently of $n$), our virtual value function satisfies Properties~\ref{item:coreprop1} and~\ref{item:coreprop2}. Moreover, as, like \cite{CaiDW16}, it is based on the utilities obtained in a second price auction, it is close enough to \cite{CaiDW16} to retain Properties~\ref{item:coreprop3},~\ref{item:coreprop4}, and~\ref{item:coreprop5}, and we can finish our proof of \autoref{thm:main}.

\paragraph{Proving \autoref{thm:mainpi}.} We show \autoref{thm:mainpi} using the same framework. The only change is that Property~\ref{item:coreprop4} needs to replaced by a prior-independent version. Namely, we want:\begin{inparaenum}[(i*)]\setcounter{enumi}{3}
\item \label{item:coreprop4star} The virtual welfare must be within a constant factor of the revenue of a (simple) prior independent auction.
\end{inparaenum} 
We show that the \iu{}-virtual welfare defined above also satisfies this property. For this, we take inspiration from \cite{GoldnerK16} and get prior-independence by using the bids of one of the bidders to get some estimate of the prior distribution. Adapting the proof in \cite{GoldnerK16} to \iu{}-virtual welfare is done in \autoref{lemma:step3alt}.

\subsection{Organization}
\label{sec:org}

For readers not familiar with duality or this line of work, we overview all the necessary definitions in \autoref{sec:prelim}. All our definitions and notations defined in \autoref{sec:prelim} are standard, so expert readers can jump to \autoref{sec:main} without losing continuity. It is in \autoref{sec:main} and \autoref{sec:cdw} that we prove our main result, and specifically, \autoref{lemma:step2} is the analogue of \autoref{lemma:boostcore} for \iu{}-virtual welfare. Finally, \autoref{app:standard} and \autoref{app:CDW} have the proofs of some standard lemmas that are used in \autoref{sec:main} and \autoref{sec:cdw}.


\section{Preliminaries} \label{sec:prelim}

We use $\reals$ to denote the set of real numbers and $\reals^{+}$ to denote the set of non-negative real numbers. For a real number $x$, we use $x^+$ to denote $\max(x,0)$. For $k, m > 0$ and $w \in \paren*{ \reals^m }^k$, we use $\max(w) \in \reals^m$ to denote the vector obtained by taking the coordinate wise maximum of $w$. That is, for all $j \in [m]$, the $j^{\text{th}}$ coordinate of $\max(w)$, which shall be denote by $\max(w)\vert_j$, is $\max(w)\vert_j = \max_{i \in [k]} w_{i,j}$.

\subsection{Probability Theory} 
\label{sec:prelimprob}

Let $\V$ be a finite set and $\D$ be a probability distribution over $\V$. Let $f : \V \to [0,1]$ denote the probability mass function of $\D$, {\em i.e.}, for all $x \in \V$, we have $f(x) = \Pr_{y \sim \D}(y = x)$. 

\paragraph{Expectation and Variance.} For a function $g : \V \to \reals$, the expectation of $g(\cdot)$ is defined as:
\[
\E_{x \sim \D}\sq*{ g(x) } = \sum_{x \in \V} f(x) \cdot g(x) .
\]
The variance of $g(\cdot)$ is:
\[
\Var_{x \sim \D}\paren*{ g(x) } = \E_{x \sim \D}\sq*{ \paren*{ g(x) - \E_{x \sim \D}\sq*{ g(x) } }^2 } = \E_{x \sim \D}\sq*{ \paren*{ g(x) }^2 } - \paren*{ \E_{x \sim \D}\sq*{ g(x) } }^2 .
\]
We omit $x \sim \D$ from the above notations when it is clear from the context. 

\paragraph{Independence.} We say two functions $g_1, g_2 : \V \to \reals$ are independent of each other if for all $a, b \in \reals$, we have
\[
\Pr_{x \sim \D}\paren*{ g_1(x) = a, g_2(x) = b } = \Pr_{x \sim \D}\paren*{ g_1(x) = a } \cdot \Pr_{x \sim \D}\paren*{ g_2(x) = b } .
\]
For $k > 0$, we say that $k$ functions $g_1, g_2, \cdots, g_k: \V \to \reals$ are (pairwise) independent if $g_i$ and $g_j$ are independent for all $i \neq j$. 

\paragraph{Standard lemmas.} The following are some standard facts and lemmas that we shall use:

\begin{fact}
\label{fact:variance} 
For any function $g(\cdot)$, we have:
\begin{enumerate}
\item \label{item:variancebound} Bounds on variance: 
\[
0 \leq \Var_{x \sim \D}\paren*{ g(x) } \leq \E_{x \sim \D}\sq*{ \paren*{ g(x) }^2 } .
\]
\item \label{item:chebyshev} Chebyshev's inequality: For all $a \geq 0$, we have
\[
\Pr_{x \sim \D}\paren*{ \abs*{ g(x) - \E\sq*{ g(x) } }  \geq a } \leq \frac{ \Var\paren*{ g(x) } }{ a^2 } .
\]
\item \label{item:varianceind} Linearity of variance assuming independence: For all $k > 0$ and all $g_1, g_2, \cdots, g_k: \V \to \reals$ that are pairwise independent, we have
\[
\Var\paren*{ \sum_{i = 1}^k g_i(x) } = \sum_{i = 1}^k \Var\paren*{ g_i(x) } .
\]
\end{enumerate}
\end{fact}

\begin{lemma}[\cite{CaiDW16}, Lemma $37$\footnote{All theorem numbers from \cite{CaiDW16} are from the arXiv version \href{https://arxiv.org/abs/1812.01577}{\tt https://arxiv.org/abs/1812.01577 }}, {\em etc.}]
\label{lemma:variance-ub}
For all functions $g: \V \to \reals^{+}$, it holds that:
\[
\Var_{x \sim \D}\paren*{ g(x) } \leq \E_{x \sim \D}\sq*{ \paren*{ g(x) }^2 } \leq 2 \cdot \paren*{ \max_{x \in \V} g(x) \cdot \Pr_{y \sim \D}\paren*{ g(y) \geq g(x) } } \cdot \paren*{ \max_{x \in \V} g(x) } .
\]
\end{lemma}
\begin{proof}
The first inequality follows from \autoref{fact:variance}, \autoref{item:variancebound}. For the second, let $a_1 < a_2 < \cdots < a_m$ be all the values of $g(x)$ when $x \in \V$ and let $a_0$ be a negative number arbitrarily close to~$0$. We have:
\begin{align*}
\E_{x \sim \D}\sq*{ \paren*{ g(x) }^2 } &= \sum_{x \in \V} f(x) \cdot \paren*{ g(x) }^2 \\
&= \sum_{i = 1}^m a_i^2 \cdot \Pr_{x \sim \D} \paren*{ g(x) = a_i } \\
&= \sum_{i = 1}^m a_i^2 \cdot \paren*{ \Pr_{x \sim \D} \paren*{ g(x) > a_{i-1} } - \Pr_{x \sim \D} \paren*{ g(x) > a_i } } \\
&= a_0^2 + \sum_{i = 1}^m \paren*{ a_i^2 - a_{i-1}^2 } \cdot \Pr_{x \sim \D} \paren*{ g(x) \geq a_i } .
\end{align*}
Using the identity $x^2 - y^2 = (x+y)(x-y)$, we get:
\begin{align*}
\E_{x \sim \D}\sq*{ \paren*{ g(x) }^2 } &\leq a_0^2 + \sum_{i = 1}^m \paren*{ a_i - a_{i-1} } \cdot 2 a_i \cdot \Pr_{x \sim \D} \paren*{ g(x) \geq a_i } \\
&\leq a_0^2 + 2 \cdot \paren*{ \max_{x \in \V} g(x) \cdot \Pr_{y \sim \D}\paren*{ g(y) \geq g(x) } } \cdot \sum_{i = 1}^m \paren*{ a_i - a_{i-1} } \\
&\leq a_0^2 + 2 \cdot \paren*{ \max_{x \in \V} g(x) \cdot \Pr_{y \sim \D}\paren*{ g(y) \geq g(x) } } \cdot \paren*{ \max_{x \in \V} g(x) - a_0 } .
\end{align*}
The lemma follows as $a_0 < 0$ was arbitrary.
\end{proof}

\subsection{Auction Design Theory} 
\label{sec:prelimauc}

The paper deals with Bayesian auction design for multiple items and multiple independent additive bidders.  Formally, this setting is defined by a tuple $\paren*{n, m, \{ \D_j \}_{j = 1}^m}$, where $n > 0$ denotes the number of bidders, $m > 0$ denotes the number of items, and $\D_j$, for $j \in [m]$, is a distribution with a finite support $\V_j \subseteq \reals^{+}$. We shall assume that, for all $j \in [m]$, all elements in $\V_j$ have non-zero probability under $\D_j$. This is without loss of generality as we can simply remove all elements that have zero probability.

We define $\D = \bigtimes_{j = 1}^m \D_j$, $\V = \bigtimes_{j = 1}^m \V_j$ and use $\D$ and $\{ \D_j \}_{j = 1}^m$ interchangeably. Bidder $i \in [n]$ has a private valuation $v_{i,j}$ for each item $j \in [m]$, that is sampled (independently for all bidders and items) from the distribution $\D_j$. We shall $v_i$ to denote the tuple $(v_{i,1}, \cdots, v_{i,m}) \in \V$, $v_{-i}$ to denote the tuple $(v_1, \cdots, v_{i-1}, v_{i+1}, \cdots v_n)$, and $v$ to denote the tuple $(v_1, \cdots, v_n) \in \V^n$. We sometimes use $(v_i, v_{-i})$ instead of $v$ if we want to emphasize the valuation of bidder $i$. 

We reserve $f_j(\cdot)$ to denote the probability mass function corresponding to $\D_j$ and $F_j(\cdot)$ to denote the cumulative mass function, {\em i.e.}, $F_j(x) = \Pr_{y \sim \D_j}(y \leq x)$. For all $k > 0$, we use $f^{(k)}$ and $F^{(k)}$ to denote the probability mass function and the cumulative mass function for the distribution $\D^k = \underbrace{\D \times \D \times \cdots \times \D}_{k \text{~times}}$. We omit $k$ when $k = 1$. We may also write $f^*$ and $F^*$ if $k$ is clear from context.

\subsubsection{Definition of an Auction}

Let $(n, m, \D)$ be an auction setting as above. For our purposes, owing to the revelation principle, it is enough to think of an auction as a pair of functions $\auc = (\pi, p)$ with the following types:
\begin{align*}
\pi &: \V^n \to \paren*{[0,1]^m}^n , \\
p &: \V^n \to \reals^n .
\end{align*}

Here, the function $\pi$ represents the `allocation function' of the auction $\auc$. It takes a tuple $v \in \V^n$ of `reported valuations' and outputs for all bidders $i \in [n]$ and items $j \in [m]$, the probability that bidder $i$ gets item $j$ when the reported types are $v$. As every item can be allocated at most once, we require for all $v \in \V^n$ and $j \in [m]$ (here, $\pi_{i,j}(v)$ denotes the $(i,j)^{\text{th}}$ coordinate of $\pi(v)$) that:
\begin{equation}
\label{eq:allocconstraint}
\sum_{i = 1}^n \pi_{i,j}(v) \leq 1 .
\end{equation}

The function $p$ denotes the `payment function' of the auction. It takes a tuple $v \in \V^n$ of reported valuations and outputs for all bidders $i \in [n]$, the amount bidder $i$ must pay the auctioneer. We shall use $p_i(v)$ to denote the $i^{\text{th}}$ coordinate of $p(v)$. 

\paragraph{The functions $\barpi(\cdot)$ and $\barp(\cdot)$.} For a bidder $i \in [n]$, we define the functions $\barpi_i : \V \to [0,1]^m$ and $\barp_i : \V \to \reals$ to be the expectation over the other bidders' valuations of the functions $\pi$ and $p$ respectively. Formally, for $v_i \in \V$, we have:
\begin{equation}
\label{eq:barpip}
\barpi_i(v_i) = \E_{v_{-i} \sim \D^{n-1}} \sq*{ \pi_i(v_i, v_{-i}) } \hspace{1cm}\text{and}\hspace{1cm} \barp_i(v_i) = \E_{v_{-i} \sim \D^{n-1}} \sq*{ p_i(v_i, v_{-i}) } .
\end{equation}

\paragraph{(Bayesian) truthfulness.} Roughly speaking, an auction is said to be truthful if the `utility' of any bidder $i \in [n]$ is maximized (and non-negative) when they report their true valuation. As the utility of bidder $i$ is defined simply as the value of player $i$ for all items allocated to them minus the payment made by player $i$, we have that an auction is truthful if for all $i \in [n], v_i, v'_i \in \V$, we have:
\begin{equation}
\label{eq:truthful}
\begin{split}
\sum_{j = 1}^m \barpi_{i,j}(v_i) \cdot v_{i,j} - \barp_i(v_i) &\geq \sum_{j = 1}^m \barpi_{i,j}(v'_i) \cdot v_{i,j} - \barp_i(v'_i) . \\
\sum_{j = 1}^m \barpi_{i,j}(v_i) \cdot v_{i,j} - \barp_i(v_i) &\geq 0 .
\end{split}
\end{equation}

Throughout this work, we restrict attention to auctions that are truthful.

\paragraph{Revenue.}
\autoref{eq:truthful} implies that one should expect the bidders in a truthful auction to report their true valuations.  When this happens, we can calculate the revenue generated by the auction $\auc$ as follows:
\begin{equation}
\label{eq:revenue}
\rev(\auc, \D, n) = \E_{v \sim \D^n} \sq*{ \sum_{i = 1}^n p_i(v) } = \sum_{i = 1}^n \E_{v_i \sim \D} \sq*{ \barp_i(v_i) } .
\end{equation}

We use $\rev(\D, n)$ to denote the maximum possible revenue of a (truthful) auction in the setting $(n, m, \D)$, {\em i.e.}, $\rev(\D, n) = \max_{\auc} \rev(\auc, \D, n)$.

\subsubsection{`Simple' Auctions}
The following well-known (truthful) auctions will be referenced throughout the proof.

\paragraph{VCG.} The VCG auction is a truthful auction that runs as follows: First, all bidders $i \in [n]$ report their valuation function $v_i$ to the auctioneer. Then, for all $j \in [m]$, item $j$ is given to the bidder with the highest (ties broken lexicographically) bid for item $j$ at a price equal to the second highest bid for item $j$. Formally, we have that for all $v \in \V^n$, $i \in [n]$, and $j \in [m]$ that:
\[
\pi_{i,j}(v) = \begin{cases}
1, &i = \argmax_{i'} v_{i',j} \\
0, &i \neq \argmax_{i'} v_{i',j} \\
\end{cases} \hspace{1cm}\text{and}\hspace{1cm}
p_i(v) = \sum_{j' : \pi_{i,j'}(v) = 1} \max_{i' \neq i} v_{i',j'} .
\]

We define $\vcg(\D, n)$ to be the revenue generated by the VCG auction.

\paragraph{Selling Separately.}  We say that an auction sells the items separately if it can be seen as $m$ separate auctions, one for each item $j$, such that the auction for item $j$ depends only on the values $\set*{ v_{i,j} }_{i \in [n]}$ the bidders have for item $j$, and the payments of the bidders is just the sum of their payments in each of the $m$ auctions. Formally, for all $j \in [m]$, there exists truthful auctions $\auc^{(j)} = (\pi^{(j)}: \V_j^n \to [0,1]^n, p^{(j)}: \V_j^n \to \reals^n)$ such that for all $v \in \V^n$, $i \in [n]$, and $j \in [m]$:
\[
\pi_{i,j}(v) = \pi_i^{(j)}\paren*{ \set*{ v_{i',j} }_{i' \in [n]} } \hspace{1cm}\text{and}\hspace{1cm} p_i(v) = \sum_{j' = 1}^m p_i^{(j')}\paren*{ \set*{ v_{i',j'} }_{i' \in [n]} } .
\]

Note the VCG auction sells the items separately. We define $\srev_j(\D_j, n)$ to be the Myerson optimal revenue for selling item $j$ to $n$ bidders. It follows that $\srev(\D, n) = \sum_{j = 1}^m \srev_j(\D_j, n)$ is the maximum revenue generated by any auction that sells the items separately

\paragraph{BVCG.} A BVCG auction is defined by a number $0 \leq k \leq n$ and\footnote{We will have $k = 0$ in the proof of \autoref{thm:main} and $k \leq 1$ in the proof of \autoref{thm:mainpi}.} and a set of non-negative numbers $\fee_{i, v_{-i}}$, for all $i \in [n-k]$ and $v_{-i} \in \V^{n-1}$. In this auction, the last $k$ bidders are treated as special and do not receive any items or pay anything. The first $n-k$ bidders participate in a VCG auction but bidder $i \in [n-k]$ only gets access to the items allocated to him in the VCG auction if he pays an entry fee $\fee_{i, v_{-i}}$ that depends on the bids $v_{-i} \in \V^{n-1}$ of all the other players in addition to the prices charged by the VCG auction.

Formally, for all $v \in \V^n$ and $j \in [m]$, we have for all $n-k < i \leq n$ that $\pi_{i,j}(v) = 0$ and $p_i(v) = 0$, and for all $i \in [n-k]$, we have:
\begin{align*}
\pi_{i,j}(v) &= \begin{cases}
1, &i = \argmax_{i' \neq i \in [n-k]} v_{i',j} \wedge \sum_{j' = 1}^m \max\paren*{ v_{i,j'} - \max_{i' \neq i \in [n-k]} v_{i',j'}, 0 } \geq \fee_{i, v_{-i}} \\
0, &\text{otherwise} \\
\end{cases} \\
p_i(v) &= \begin{cases}
\fee_{i, v_{-i}} + \sum_{j' : \pi_{i,j'}(v) = 1} \max_{i' \neq i \in [n-k]} v_{i',j'}, &\text{~if~} \exists j' : \pi_{i,j'}(v) = 1 \\
0, &\text{otherwise} \\
\end{cases}
\end{align*}

We define $\bvcg(\D, n)$ to be the maximum revenue of a BVCG auction with $n$ bidders. Also, define $\pibvcg(\D, n)$ to be the maximum revenue of a prior-independent BVCG auction with $n$ bidders, {\em i.e,} where the values $\fee_{i, v_{-i}}$ for all $i$ and $v_{-i}$ are not a function of the distribution $\D$.

\subsubsection{Myerson's Virtual Values}
We define the virtual value function following \cite{CaiDW16}. Throughout this subsection, we fix our attention on a single item $j \in [m]$ in an auction setting. The notations $\D_j$, $\V_j$, $f_j$, $F_j$ will be the same as above. Recall our assumption that $f_j(x) > 0$ for all $x \in \V_j$. 

\begin{definition}
\label{def:vv}
The virtual value function $\vv_j : \V_j \to \reals$ is defined to be:
\[
\vv_j(x) = \begin{cases}
x, &\text{~if~} x = \max(\V_j) \\
x - \frac{ \paren*{ x' - x } \cdot \paren*{ 1 - F_j(x) } }{f_j(x)}, &\text{~if~} x \neq \max(\V_j) 
\end{cases} .
\]
Here, $x'$ denotes the smallest element $>x$ in $\V_j$ and is well defined for all $x \neq \max(\V_j)$.
\end{definition}

Using the function $\vv_j(\cdot)$, one can compute the ironed virtual value function as described in \autoref{algo:ivv}.

\begin{algorithm}
\caption{Computing the ironed virtual value function $\ivv_j : \V_j \to \reals$.}
\label{algo:ivv}
\begin{algorithmic}[1]

\State $x \gets \max(\V_j)$. 

\While{$\mathtt{True}$}

\State For all $y \in \V_j, y \leq x$, set:
\[
a(y) = \frac{ \sum_{y' \in [y, x] \cap \V_j} f_j(y') \cdot \vv_j(y') }{ \sum_{y' \in [y, x] \cap \V_j} f_j(y') } . \label{line:seta}
\]

\State $y^* \gets \argmax_{y \in \V_j, y \leq x} a(y)$ breaking ties in favor of larger values. \label{line:setystar}

\State For all $y' \in [y^*, x] \cap \V_j$, set $\ivv_j(y') \gets a(y^*)$.  \label{line:setivv}

\State If $y^* = \min(\V_j)$, abort. Else, $x \gets $ the largest element $<y^*$ in $\V_j$. \label{line:abort}

\EndWhile

\end{algorithmic}
\end{algorithm}

The ironed virtual value function $\ivv_j(\cdot)$ has several nice properties some of which recall below. The lemmas are adapted from \cite{CaiDW16} which also has a more in depth discussion.

\begin{lemma}
\label{lemma:ivvmonotone}
For all $x, x' \in \V_j$ such that $x \leq x'$, we have $\ivv_j(x) \leq \ivv_j(x')$. 
\end{lemma}
\begin{proof}
As $x \leq x'$, \autoref{algo:ivv} did not set the value of $\ivv_j(x')$ after setting the value of $\ivv_j(x)$. Using this and \autoref{line:setivv} of \autoref{algo:ivv}, we get that it is sufficient to show that the value $a(y^*)$ cannot increase between two consecutive iterations of the While loop. To this end, consider two consecutive iterations and let $x_1, y^*_1, a_1(\cdot)$ and $x_2, y^*_2, a_2(\cdot)$ be the values of the corresponding variables in the first and the second iteration respectively and note that $y^*_2 \leq x_2 < y^*_1 \leq x_1$. 

By our choice of $y^*_1$ in \autoref{line:setystar} in the first iteration, we have that $a_1(y^*_1) \geq a_1(y^*_2)$. Extending using \autoref{line:seta}, we get:
\begin{align*}
a_1(y^*_1) \geq a_1(y^*_2) &= \frac{ \sum_{y' \in [y^*_2, x_1] \cap \V_j} f_j(y') \cdot \vv_j(y') }{ \sum_{y' \in [y^*_2, x_1] \cap \V_j} f_j(y') } \\
&= \frac{ \sum_{y' \in [y^*_2, x_2] \cap \V_j} f_j(y') }{ \sum_{y' \in [y^*_2, x_1] \cap \V_j} f_j(y') } \cdot a_2(y^*_2) + \frac{ \sum_{y' \in [y^*_1, x_1] \cap \V_j} f_j(y') }{ \sum_{y' \in [y^*_2, x_1] \cap \V_j} f_j(y') } \cdot a_1(y^*_1) .
\end{align*}
It follows that $a_1(y^*_1) \geq a_2(y^*_2)$, as desired.

\end{proof}

\begin{lemma}
\label{lemma:ivv-ub}
For all $x \in \V_j$, we have $\ivv_j(x) \leq x$. 
\end{lemma}
\begin{proof}
Let $x_1, y^*_1, a_1(\cdot)$ be the values of the corresponding variables in the iteration when the value of $\ivv_j(x)$ is set. Observe that $y^*_1 \leq x \leq x_1$. If $x = x_1$, we simply have:
\[
\ivv_j(x) = a_1(y^*_1) = \frac{ \sum_{y' \in [y^*_1, x_1] \cap \V_j} f_j(y') \cdot \vv_j(y') }{ \sum_{y' \in [y^*_1, x_1] \cap \V_j} f_j(y') } \leq \frac{ \sum_{y' \in [y^*_1, x_1] \cap \V_j} f_j(y') \cdot x }{ \sum_{y' \in [y^*_1, x_1] \cap \V_j} f_j(y') } = x ,
\]
where the penultimate step uses $\vv(y') \leq y' \leq x_1 = x$ by \autoref{def:vv}. Otherwise, we have $x < x_1$. Define $x' \in \V_j$ to be the smallest such that $x < x'$ and observe that $x' \leq x_1$. By our choice of $y^*_1$ in \autoref{line:setystar}, we have:
\begin{align*}
a_1(y^*_1) &\geq a_1(x') \\
&= \frac{ \sum_{y' \in [x', x_1] \cap \V_j} f_j(y') \cdot \vv_j(y') }{ \sum_{y' \in [x', x_1] \cap \V_j} f_j(y') } \\
&= a_1(y^*_1) \cdot \frac{ \sum_{y' \in [y^*_1, x_1] \cap \V_j} f_j(y') }{ \sum_{y' \in [x', x_1] \cap \V_j} f_j(y') } - \frac{ \sum_{y' \in [y^*_1, x] \cap \V_j} f_j(y') \cdot \vv_j(y') }{ \sum_{y' \in [y^*_1, x] \cap \V_j} f_j(y') } \cdot \frac{ \sum_{y' \in [y^*_1, x] \cap \V_j} f_j(y') }{ \sum_{y' \in [x', x_1] \cap \V_j} f_j(y') } .
\end{align*}
Rearranging, we get:
\[
\ivv_j(x) = a_1(y^*_1) \leq \frac{ \sum_{y' \in [y^*_1, x] \cap \V_j} f_j(y') \cdot \vv_j(y') }{ \sum_{y' \in [y^*_1, x] \cap \V_j} f_j(y') } \leq \frac{ \sum_{y' \in [y^*_1, x] \cap \V_j} f_j(y') \cdot x }{ \sum_{y' \in [y^*_1, x] \cap \V_j} f_j(y') } = x ,
\]
using $\vv(y') \leq y' \leq x$ by \autoref{def:vv} in the penultimate step.
\end{proof}

Myerson \cite{Myerson81} proved that when there is only one item and $\D$ is continuous, the optimal revenue is equal to the expected (Myerson's) virtual welfare. \cite{CaiDW16} shows that Myerson's lemma also applies to $\D$ that are discrete. For the $m$ item setting considered in this paper, we get:

\begin{proposition}[Myerson's Lemma~\cite{Myerson81, CaiDW16}]
\label{prop:myerson}
It holds that:
\[
\srev(\D, n) = \sum_{j=1}^m \srev_j(\D_j, n) = \sum_{j=1}^m \E_{v \sim \D^n} \sq*{\max_{i\in[n]} \paren*{\ivv_{j}(v_{i, j})}^+} .
\]
\end{proposition}

\begin{definition}[Regular Distributions]
\label{def:regular} 
The distribution $\D_j$ is called regular when $\vv_j(\cdot) = \ivv_j(\cdot)$, or equivalently, when $\vv_j(\cdot)$ is monotone increasing. Namely, for all $x, x' \in \V_j$ such that $x \leq x'$, we have $\vv_j(x) \leq \vv_j(x')$.
\end{definition}
In the seminal paper~\cite{BulowK96}, Bulow and Klemperer show that the maximum possible revenue from an auction selling one item to any number $n$ of bidders whose valuations for the item are sampled from the same regular distribution is at most the revenue of the (simple and prior-independent) VCG auction with $n+1$ bidders. It follows that, when there are multiple items, we have:

\begin{proposition}[Classic Bulow-Klemperer \cite{BulowK96}] 
\label{prop:bk}
If $\D$ is a product of regular distributions, $\srev(\D, n) \leq \vcg(\D, n + 1)$.
\end{proposition}


\section{Proof of Main Result} \label{sec:main} 

This section formally states and proves our main results \autoref{thm:main} and \autoref{thm:mainpi}. Using the notation developed in \autoref{sec:prelim}, we can rewrite \autoref{thm:main} as:
\begin{theorem}[Formal statement of \autoref{thm:main} and \autoref{thm:mainpi}]
\label{thm:mainformal}
Let $(n, m, \D)$ be an auction setting as in \autoref{sec:prelimauc}. Let $\epsilon > 0$ and define $n' = 20n/\epsilon$. If $\paren*{ 1- \epsilon } \cdot \rev(\D, n) > \vcg(\D, n')$, we have that:
\[
\rev(\D, n) \leq \max\paren*{ \bvcg(\D, n'), \srev(\D, n') } .
\]
Furthermore, if $\D$ is a product of regular distributions, the same assumption also implies $\rev(\D, n) \leq \pibvcg(\D, n' + 1)$.
\end{theorem}

The proof of \autoref{thm:mainformal} spans the rest of this section. We fix an auction setting $(n, m, \D)$ and $\epsilon > 0$. To simplify notation, we drop $m, \D$ from the arguments but retain $n$ when we want to emphasize the number of bidders. Our proof has three main steps: First, we use the duality framework of \cite{CaiDW16} to get a suitable upper bound on $\rev(n)$. As explained in \autoref{sec:overviewmultiplebidder}, the ``standard'' duality framework seems to be insufficient for our needs, and out first step is to show a new duality benchmark, called the independent utilities, or the $\iu(n')$-benchmark.

To define our benchmark, we first define the $\iu(n')$-virtual value of bidder $i \in [n]$. Let $v_i$ be the valuation (or the type) of bidder $i$ and let $w_{-i}$ be the valuations of $n' - 1$ ``ghost'' bidders. We first partition the set $\V$ of all possible valuations for bidder $i$ into $m + 1$ regions based on $w_{-i}$. Namely, we define a region $\region^{(n')}_j(w_{-i})$ for each item $j \in [m]$ and also a region $\region^{(n')}_0(w_{-i}) = \V \setminus \bigcup_{j \in [m]} \region^{(n')}_j(w_{-i})$ of all elements of $\V$ that are not in any of the $\region^{(n')}_j(w_{-i})$. For $j \in [m]$, we say that $v_i \in \region^{(n')}_j(w_{-i})$ if:
\begin{equation}
\label{eq:region}
\region^{(n')}_j(w_{-i}) = \set*{ v_i \in \V ~\middle\vert~ v_{i,j} \geq \max(w_{-i})\vert_j \wedge j \text{~smallest in~} \argmax_{j'} \paren*{ v_{i,j'} - \max(w_{-i})\vert_{j'} } } .
\end{equation}

Having defined these regions for each $w_{-i}$, we next consider the probability, for all $j \in [m]$, that the valuation $v_i \in \region^{(n')}_j(w_{-i})$ where the probability is over choice of the types $w_{-i}$ of the ghost bidders. Formally, 
\begin{equation}
\label{eq:pregion}
\pregion^{(n')}_j(v_i) = \Pr_{w_{-i} \sim \D^{n'-1}}\paren*{ v_i \in \region^{(n')}_j(w_{-i}) } .
\end{equation}
Next, for $j \in [m]$, define the $\iu(n')$-virtual value of bidder $i$ for item $j$ as:
\begin{equation}
\label{eq:iuvv}
\fv^{(n')}_j(v_i) = v_{i,j} \cdot \paren*{ 1 - \pregion^{(n')}_j(v_i) } + \ivv_j(v_{i,j})^+ \cdot \pregion^{(n')}_j(v_i) .
\end{equation}
We drop the superscript $n'$ from all the above notations when it is clear from context. We now define the $\iu(n')$-benchmark as:
\begin{equation}
\label{eq:iu}
\ub(n, n') = \sum_{j = 1}^m \E_{v \sim \D^n} \sq*{ \max_{i \in [n]} \set*{ \fv^{(n')}_j(v_i) } } .
\end{equation}
The first step in our proof is to show that the $\iu(n')$-benchmark indeed upper bounds $\rev(n)$. 
\begin{lemma}
\label{lemma:step1}
It holds that:
\[
\rev(n) \leq \ub(n, n') .
\]
\end{lemma}

 The next step of the proof is to show that, under the assumption of the theorem, our $\iu(n')$-benchmark increases (significantly) as the number of bidders increases. We have:
\begin{lemma}
\label{lemma:step2} 
Assume that $\paren*{ 1- \epsilon } \cdot \rev(n) > \vcg(n')$. We have:
\[ 
\ub(n, n') \leq \frac{1}{20} \cdot \ub(n', n')
\] 
\end{lemma}
We mention that the choice of the constant $20$ in \autoref{lemma:step2} is arbitrary and it can be replaced by any other value as long as the value of $n'$ is changed accordingly. In fact, it can even be function of $\epsilon$. As the last step of the proof, we show the following upper bounds on $\ub(n', n')$.

\begin{lemma}
\label{lemma:step3}
For all $n'' \leq n'$, we have:
\[
\ub(n'', n') \leq 2 \cdot \bvcg(n') + 6 \cdot \srev(n') .
\]
\end{lemma}

\begin{lemma}
\label{lemma:step3alt}
If $\D$ is a product of regular distributions, then, for all $n'' \leq n'$, we have:
\[
\ub(n'', n') \leq 17 \cdot \pibvcg(n' + 1). 
\]
\end{lemma}

Once we have these lemmas, \autoref{thm:mainformal} is almost direct. We include the two-line proof below.

\begin{proof}[Proof of \autoref{thm:mainformal}]
By \autoref{lemma:step1} and \autoref{lemma:step2}, we have that
\[
\rev(n) \leq \ub(n, n') \leq \frac{1}{20} \cdot \ub(n', n') .
\]
Next, using \autoref{lemma:step3} and also \autoref{lemma:step3alt} for the ``furthermore'' part, we have:
\[
\rev(n) \leq \max\paren*{ \bvcg(\D, n'), \srev(\D, n') } \hspace{1cm}\text{and}\hspace{1cm} \rev(n) \leq \pibvcg(\D, n' + 1) .
\]
\end{proof}

\subsection{Proof of \autoref{lemma:step1}}

\begin{proof}
Let $\auc$ be the auction that maximizes revenue amongst all (Bayesian) truthful auctions, and let $(\barpi, \barp)$ be as defined \autoref{sec:prelimauc}. By definition, we have that $\rev(n) = \sum_{i = 1}^n \E_{v_i \sim \D} \sq*{ \barp_i(v_i) } $.

As the partition of $\V$ defined by $\set{ \region^{(n')}_j(w_{-i}) }_{j \in \set*{0} \cup [m]}$ is ``upwards closed''\footnote{A partition $\set{ R_j }_{j \in \set*{0} \cup [m]}$ is said to be upwards closed if for all $j \in [m]$, $v \in R_j$, and all $\alpha > 0$, if $v + \alpha \cdot e_j \in \V$, then $v_i + \alpha \cdot e_j \in R_j$ (Here, $e_j$ means the $j^{\text{th}}$ standard basis of $\reals^m$). } for all $w_{-i} \in \V^{n' - 1}$, we have from Corollary $27$ of \cite{CaiDW16} that\footnote{To make this paper self-contained, we recall their proof of Corollary $27$ as \autoref{lemma:cor27} in \autoref{app:CDW}.}, for all $w_{-i}$:
\begin{align*}
\rev(n) &= \sum_{i = 1}^n \E_{v_i \sim \D} \sq*{ \barp_i(v_i) } \\
&\leq \sum_{i = 1}^n \sum_{j = 1}^m \E_{v_i \sim \D} \sq*{ \barpi_{i,j}(v_i) \cdot \paren*{ v_{i,j} \cdot \mathds{1}\paren*{ v_i \notin \region^{(n')}_j(w_{-i}) } + \ivv_j(v_{i,j})^+ \cdot \mathds{1}\paren*{ v_i \in \region^{(n')}_j(w_{-i}) } } } .
\end{align*}

Therefore, by a weighted sum of the above equation over $w_{-i} \in \V^{n' - 1}$, 
\begin{align*}
\rev(n) &\leq \sum_{i = 1}^n \sum_{j = 1}^m \E_{v_i \sim \D} \sq*{ \barpi_{i,j}(v_i) \cdot \paren*{ v_{i,j} \cdot \paren*{ 1 - \pregion^{(n')}_j(v_i) } + \ivv_j(v_{i,j})^+ \cdot \pregion^{(n')}_j(v_i) } } \tag{\autoref{eq:pregion}} \\
&\leq \sum_{i = 1}^n \sum_{j = 1}^m \E_{v_i \sim \D} \sq*{ \barpi_{i,j}(v_i) \cdot \fv^{(n')}_j(v_i) } \tag{\autoref{eq:iuvv}} .
\end{align*}

Recall that the function $\barpi(v)$ is just the expected value (over the randomness in $v_{-i}$) of $\pi_i(v_i, v_{-i})$. Using this and the fact that $\sum_{i = 1}^n \pi_{i,j}(v) \leq 1$ (from \autoref{eq:allocconstraint}), we have that:
\begin{align*}
\rev(n) &\leq \sum_{i = 1}^n \sum_{j = 1}^m \sum_{v_i \in \V} f(v_i) \cdot \paren*{ \sum_{v_{-i} \in \V^{n-1}} f^*(v_{-i}) \cdot \pi_{i,j}(v_i, v_{-i}) } \cdot \fv^{(n')}_j(v_i) \\
&\leq \sum_{i = 1}^n \sum_{j = 1}^m \sum_{v \in \V^n} f^*(v) \cdot \pi_{i,j}(v_i, v_{-i}) \cdot \fv^{(n')}_j(v_i) \\
&\leq \sum_{j = 1}^m \sum_{v \in \V^n} f^*(v) \cdot \max_{i \in [n]}\set*{ \fv^{(n')}_j(v_i) } \tag{As $\sum_{i = 1}^n \pi_{i,j}(v) \leq 1$ and $\fv^{(n')}_j(v_i) \geq 0$} \\
&\leq \sum_{j = 1}^m \E_{v \sim \D^n} \sq*{ \max_{i \in [n]} \set*{ \fv^{(n')}_j(v_i) } } \\
&= \ub(n, n') \tag{\autoref{eq:iu}} .
\end{align*}
\end{proof}

\subsection{Proof of \autoref{lemma:step2}}

We start by showing a technical lemma saying that if $D$ is a distribution over a finite set $S \subseteq \reals$ and $X = x_1, \cdots, x_k$ is a $k$-length vector of independent samples from $D$, the maximum value $x_i$ is independent of the maximizer $i$ (with the ``correct'' tie-braking rule). The tie breaking rule we use is the randomized tie breaking rule $\maxind$ satisfying, for all $i \in [k]$, that:
\begin{equation}
\label{eq:maxind}
\Pr\paren*{ \maxind(X) = i } = \begin{cases}
0, &\text{~if~} i \notin \argmax_{i'} \set*{ x_{i'} } \\
\frac{1}{ \card*{ \argmax_{i'} \set*{ x_{i'} } } }, &\text{~if~} i \in \argmax_{i'} \set*{ x_{i'} } 
\end{cases}
\end{equation}
It holds that:
\begin{lemma}
\label{lemma:maxind}
Let $k > 0$, $S \subseteq \reals$ be a finite set, and $D$ be a distribution over $S$. For all $s \in S$ and $i \in [k]$, we have that
\[
\Pr_{ X \sim D^k }\paren*{ \max_{i'} x_{i'} = s \mid \maxind(X) = i } = \Pr_{ X \sim D^k }\paren*{ \max_{i'} x_{i'} = s } .
\]
\end{lemma}
\begin{proof}
We omit the subscript $X \sim D^k$ to keep the notation concise. Observe that we have $\Pr\paren*{ \maxind(X) = i } = \frac{1}{k}$ by symmetry and also that:
\[
\Pr\paren*{ \max_{i'} x_{i'} = s } = \paren*{ \Pr_{x \sim D}\paren*{ x \leq s } }^k - \paren*{ \Pr_{x \sim D}\paren*{ x < s } }^k .
\]
This means that it is sufficient to show that:
\[
\Pr\paren*{ \max_{i'} x_{i'} = s \wedge \maxind(X) = i } = \frac{1}{k} \cdot \paren*{ \paren*{ \Pr_{x \sim D}\paren*{ x \leq s } }^k - \paren*{ \Pr_{x \sim D}\paren*{ x < s } }^k } .
\]
We show this by considering all possible values of $\argmax_{i'} x_{i'}$. We have:
\begin{align*}
\Pr\paren*{ \max_{i'} x_{i'} = s \wedge \maxind(X) = i } &= \sum_{S \ni i} \Pr\paren*{ \max_{i'} x_{i'} = s \wedge \maxind(X) = i \wedge \argmax_{i'} x_{i'} = S } \\
&= \sum_{S \ni i} \frac{1}{\card*{S}} \cdot \Pr\paren*{ \max_{i'} x_{i'} = s \wedge \argmax_{i'} x_{i'} = S } \tag{\autoref{eq:maxind}} .
\end{align*}
We can calculate the term on the right:
\begin{align*}
\Pr\paren*{ \max_{i'} x_{i'} = s \wedge \maxind(X) = i } &= \sum_{S \ni i} \frac{1}{\card*{S}} \cdot \paren*{ \Pr_{x \sim D}\paren*{ x = s } }^{\card*{S}} \cdot \paren*{ \Pr_{x \sim D}\paren*{ x < s } }^{ k - \card*{S} } \\
&= \sum_{k' = 1}^k \frac{1}{k'} \cdot \binom{k - 1}{k' - 1} \cdot \paren*{ \Pr_{x \sim D}\paren*{ x = s } }^{k'} \cdot \paren*{ \Pr_{x \sim D}\paren*{ x < s } }^{ k - k' } \\
&= \sum_{k' = 1}^k \frac{1}{k} \cdot \binom{k}{k'} \cdot \paren*{ \Pr_{x \sim D}\paren*{ x = s } }^{k'} \cdot \paren*{ \Pr_{x \sim D}\paren*{ x < s } }^{ k - k' } .
\end{align*}
The Binomial theorem $(a+b)^n = \sum_{i = 0}^n \binom{n}{i} a^i b^{n-i}$ then gives:
\[
\Pr\paren*{ \max_{i'} x_{i'} = s \wedge \maxind(X) = i } = \frac{1}{k} \cdot \paren*{ \paren*{ \Pr_{x \sim D}\paren*{ x \leq s } }^k - \paren*{ \Pr_{x \sim D}\paren*{ x < s } }^k } .
\]

\end{proof}

We now prove \autoref{lemma:step2}.

\begin{proof}[Proof of \autoref{lemma:step2}]
Recall from the definition of $\ub(n, n')$ in \autoref{eq:iu} that $\ub(n, n') = \sum_{j = 1}^m \E_{v \sim \D^n} \sq*{ \max_{i \in [n]} \set*{ \fv^{(n')}_j(v_i) } }$. As we can always sample values for more bidders and not use them, we also have:
\[
\ub(n, n') = \sum_{j = 1}^m \E_{v \sim \D^{n'}} \sq*{ \max_{i \in [n]} \set*{ \fv^{(n')}_j(v_i) } } .
\]
For $v \in \D^{n'}$, we use $\fv^{(n')}_j(v)$ to denote the vector $\paren*{ \fv^{(n')}_j(v_1), \cdots, \fv^{(n')}_j(v_{n'}) }$. Over the random space defined by the distribution $v \sim D^{n'}$, define the event $\mathcal{E} := \maxind\paren*{ \fv^{(n')}_j(v) } \in [n]$. We get:
\begin{align*}
\ub(n, n') &= \sum_{j = 1}^m \Pr_{v \sim \D^{n'}}\paren*{ \mathcal{E} } \cdot \E_{v \sim \D^{n'}} \sq*{ \max_{i \in [n]} \set*{ \fv^{(n')}_j(v_i) } \mid \mathcal{E} } + \Pr_{v \sim \D^{n'}}\paren*{ \overline{\mathcal{E}} } \cdot \E_{v \sim \D^{n'}} \sq*{ \max_{i \in [n]} \set*{ \fv^{(n')}_j(v_i) } \mid \overline{\mathcal{E}} } \\
&= \sum_{j = 1}^m \frac{\epsilon}{20} \cdot \E_{v \sim \D^{n'}} \sq*{ \max_{i \in [n]} \set*{ \fv^{(n')}_j(v_i) } \mid \mathcal{E} } + \Pr_{v \sim \D^{n'}}\paren*{ \overline{\mathcal{E}} } \cdot \E_{v \sim \D^{n'}} \sq*{ \max_{i \in [n]} \set*{ \fv^{(n')}_j(v_i) } \mid \overline{\mathcal{E}} } \tag{As $\Pr_{v \sim \D^{n'}}\paren*{ \mathcal{E} } = \frac{n}{n'} = \frac{\epsilon}{20}$} .
\end{align*}
Next, note that the maximum over the first $n$ coordinates is at most the maximum over all the coordinates. Moreover, conditioned on $\overline{\mathcal{E}}$, it is at most the second highest value over all the coordinates. Using $\snd(\cdot)$ denote the second largest value in a set, we get that:
\[
\ub(n, n') \leq \sum_{j = 1}^m \frac{\epsilon}{20} \cdot \E_{v \sim \D^{n'}} \sq*{ \max_{i \in [n']} \set*{ \fv^{(n')}_j(v_i) } \mid \mathcal{E} } + \Pr_{v \sim \D^{n'}}\paren*{ \overline{\mathcal{E}} } \cdot \E_{v \sim \D^{n'}} \sq*{ \snd_{i \in [n']} \set*{ \fv^{(n')}_j(v_i) } \mid \overline{\mathcal{E}} } .
\]
From \autoref{lemma:ivv-ub}, we have $\ivv_j(v_{i,j}) \leq v_{i,j}$, we also have $\fv^{(n')}_j(v_i) \leq v_{i,j}$ for any $n'$. This means that the second highest value of $\set*{ \fv^{(n')}_j(v_i) }$ is at most the second highest value of $\set*{ v_{i,j} }$. We get:
\[
\ub(n, n') \leq \sum_{j = 1}^m \frac{\epsilon}{20} \cdot \E_{v \sim \D^{n'}} \sq*{ \max_{i \in [n']} \set*{ \fv^{(n')}_j(v_i) } \mid \mathcal{E} } + \Pr_{v \sim \D^{n'}}\paren*{ \overline{\mathcal{E}} } \cdot \E_{v \sim \D^{n'}} \sq*{ \snd_{i \in [n']} \set*{ v_{i,j} } \mid \overline{\mathcal{E}} } .
\]
To continue, note that all the second highest value of $\set*{ v_{i,j} }$ is non-negative. We get:
\[
\ub(n, n') \leq \sum_{j = 1}^m \frac{\epsilon}{20} \cdot \E_{v \sim \D^{n'}} \sq*{ \max_{i \in [n']} \set*{ \fv^{(n')}_j(v_i) } \mid \mathcal{E} } + \E_{v \sim \D^{n'}} \sq*{ \snd_{i \in [n']} \set*{ v_{i,j} } } .
\]
By definition, the second term is simply $\vcg(n')$. We remove the conditioning from the first term using \autoref{lemma:maxind}. We get:
\[
\ub(n, n') \leq \sum_{j = 1}^m \frac{\epsilon}{20} \cdot \E_{v \sim \D^{n'}} \sq*{ \max_{i \in [n']} \set*{ \fv^{(n')}_j(v_i) } } + \vcg(n') .
\]
From \autoref{eq:iu}, we conclude:
\[
\ub(n, n') \leq \frac{\epsilon}{20} \cdot \ub(n', n') + \vcg(n') .
\]
Under the assumption in \autoref{lemma:step2}, we have that $\paren*{ 1- \epsilon } \cdot \rev(n) > \vcg(n')$. From \autoref{lemma:step1}, we have $\rev(n) \leq \ub(n, n')$. Plugging in, we get:
\[
\ub(n, n') \leq \frac{\epsilon}{20} \cdot \ub(n', n') + \paren*{ 1- \epsilon } \cdot \ub(n, n') .
\]
Rearranging gives the lemma.

\end{proof}


\section{Proofs of \autoref{lemma:step3} and \autoref{lemma:step3alt}}
\label{sec:cdw}

We show \autoref{lemma:step3} and \autoref{lemma:step3alt} following the framework of \cite{CaiDW16}. The first step is common to both the lemmas and shows that $\ub(n'', n')$ is at most $4 \cdot \srev(n')$ plus an additional term corresponding to the term \core{} in \cite{CaiDW16}. This is captured in \autoref{lemma:step3step1}. The next step bounds \core{} in two different ways to show the two lemmas. These can be found in \autoref{sec:step3step2} and \autoref{sec:step3step2alt}.

\subsection{Step $1$ -- Decomposing $\ub(\cdot)$}
\label{sec:step3:decompose}

\begin{lemma}
\label{lemma:step3step1} For all $n'' \leq n'$, we have:
\[
\ub(n'', n') \leq 4 \cdot \srev(n') + \core ,
\]
where we define \core{} as:
\begin{align*}
r^*_{\ron, j}(x) &= \max_{y > x} y \cdot \Pr_{y' \sim \D_j}\paren*{ y' \geq y } &\forall& j \in [m] . \\
r^{(i)}_{\ron}(w_{-i}) &= \sum_{j = 1}^m r^*_{\ron, j}(\max(w_{-i})\vert_j) &\forall& i \in [n'], w_{-i} \in \V^{n'-1} . \\
\T_{i,j}(w_{-i}) &= r^{(i)}_{\ron}(w_{-i}) + \max(w_{-i})\vert_j &\forall& j \in [m], i \in [n'], w_{-i} \in \V^{n'-1} . 
\end{align*}
\[
\core = \sum_{j = 1}^m \sum_{i = 1}^{n'} \sum_{w_{-i}} \sum_{ \max(w_{-i})\vert_j \leq v_{i,j} \leq \T_{i,j}(w_{-i}) } f^*(w_{-i}) f_j(v_{i,j}) \cdot \paren*{ v_{i,j} - \max(w_{-i})\vert_j } .
\]
\end{lemma}
\begin{proof}
Fix $n'' \leq n'$. We first get rid of the parameter $n''$ by showing that $n'' = n'$ is the hardest case for the lemma. We have:
\begin{align*}
\ub(n'', n') &= \sum_{j = 1}^m \E_{v \sim \D^{n''}} \sq*{ \max_{i \in [n'']} \set*{ v_{i,j} \cdot \paren*{ 1 - \pregion_j(v_i) } + \ivv_j(v_{i,j})^+ \cdot \pregion_j(v_i) } } \\
&\leq \sum_{j = 1}^m \E_{v \sim \D^{n'}} \sq*{ \max_{i \in [n']} \set*{ v_{i,j} \cdot \paren*{ 1 - \pregion_j(v_i) } + \ivv_j(v_{i,j})^+ \cdot \pregion_j(v_i) } } \\
&= \ub(n') .
\end{align*}

Henceforth, we focus on upper bounding $\ub(n')$. Note that the term $1 - \pregion_j(v_i)$ in $\ub(n')$ corresponds to the event that $v_i \notin \region_j(w_{-i})$. By our choice of the regions $\region_j(\cdot)$, whenever this happens, either $v_i$ is less than $\max(w_{-i})\vert_j$ or there is a $j' \neq j$ such that the utility from $j'$ is at least as much as that from $j$. To capture these cases we define the sets\footnote{For readers familiar with \cite{CaiDW16}, our naming of these events corresponds to that used in \cite{CaiDW16}, {\em e.g.}, \nf{} corresponds to \nonfavorite{}.}:
\begin{equation}
\label{eq:events}
\begin{split}
E^{\und}_j(v_i) &= w_{-i} \in \V^{n'-1} \mid v_{i,j} < \max(w_{-i})\vert_j \\
E^{\srp}_j(v_i) &= w_{-i} \in \V^{n'-1} \mid \exists j' \neq j : v_{i,j} - \max(w_{-i})\vert_j \leq v_{i,j'} - \max(w_{-i})\vert_{j'} \\
E^{\nf}_j(v_i) &= E^{\srp}_j(v_i) \setminus E^{\und}_j(v_i) .
\end{split}
\end{equation}
As mentioned before, when $v_i \notin \region_j(w_{-i})$, we either have $w_{-i} \in E^{\und}_j(v_i)$ or $w_{-i} \in E^{\nf}_j(v_i)$. Thus, we have the following inequality.
\begin{equation}
\label{eq:eventunion}
1 - \pregion_j(v_i) \leq \Pr_{w_{-i} \sim \D^{n'-1}}\paren*{ w_{-i} \in E^{\und}_j(v_i) } + \Pr_{w_{-i} \sim \D^{n'-1}}\paren*{ w_{-i} \in E^{\nf}_j(v_i) } .
\end{equation}

Using \autoref{eq:eventunion}, we decompose $\ub(n')$ as follows:
\begin{align*}
\ub(n') &= \sum_{j = 1}^m \E_{v \sim \D^{n'}} \sq*{ \max_{i \in [n']} \set*{ v_{i,j} \cdot \paren*{ 1 - \pregion_j(v_i) } + \ivv_j(v_{i,j})^+ \cdot \pregion_j(v_i) } } \\
&\leq \sum_{j = 1}^m \E_{v \sim \D^{n'}} \sq*{ \max_{i \in [n']} \set*{ v_{i,j} \cdot \paren*{ 1 - \pregion_j(v_i) } } } + \sum_{j = 1}^m \E_{v \sim \D^{n'}} \sq*{ \max_{i \in [n']} \set*{ \ivv_j(v_{i,j})^+ \cdot \pregion_j(v_i) } } \\
&\leq \sum_{j = 1}^m \E_{v \sim \D^{n'}} \sq*{ \max_{i \in [n']} \set*{ \ivv_j(v_{i,j})^+ \cdot \pregion_j(v_i) } } \tag{\single{}}\\
&\hspace{0.5cm}+ \sum_{j = 1}^m \E_{v \sim \D^{n'}} \sq*{ \max_{i \in [n']} \set*{ v_{i,j} \cdot \Pr_{w_{-i}}\paren*{ w_{-i} \in E^{\und}_j(v_i) } } } \tag{\under{}}\\
&\hspace{0.5cm}+ \sum_{j = 1}^m \E_{v \sim \D^{n'}} \sq*{ \max_{i \in [n']} \set*{ v_{i,j} \cdot \Pr_{w_{-i}}\paren*{ w_{-i} \in E^{\nf}_j(v_i) } } } \tag{\nonfavorite{}} .
\end{align*}
We have now split $\ub(n')$ into three terms, \single{}, \under{}, and \nonfavorite{}. We will later show that both \single{} and \under{} are at most $\srev(n')$. As far as the term \nonfavorite{} goes, we need to decompose it further.
We have:
\begin{align*}
\nonfavorite &= \sum_{j = 1}^m \E_{v \sim \D^{n'}} \sq*{ \max_{i \in [n']} \set*{ \sum_{w_{-i}} f^*(w_{-i}) \cdot v_{i,j} \cdot \mathds{1}\paren*{ w_{-i} \in E^{\nf}_j(v_i) } } } \\
&\leq \sum_{j = 1}^m \E_{v \sim \D^{n'}} \sq*{ \max_{i \in [n']} \set*{ \sum_{w_{-i}} f^*(w_{-i}) \cdot \paren*{ v_{i,j} - \max(w_{-i})\vert_j } \cdot \mathds{1}\paren*{ w_{-i} \in E^{\nf}_j(v_i) } } } \\
&\hspace{0.5cm}+ \sum_{j = 1}^m \E_{v \sim \D^{n'}} \sq*{ \max_{i \in [n']} \set*{ \sum_{w_{-i}} f^*(w_{-i}) \cdot \max(w_{-i})\vert_j \cdot \mathds{1}\paren*{ w_{-i} \in E^{\nf}_j(v_i) } } } .
\end{align*}
Plugging into the previous decomposition and using the fact that $E^{\nf}_j(v_i)$ and $E^{\und}_j(v_i)$ are disjoint by \autoref{eq:events}, we get that:
\begin{align*}
\ub(n') &\leq \single + \under \\
&\hspace{0.5cm}+ \sum_{j = 1}^m \E_{v \sim \D^{n'}} \sq*{ \max_{i \in [n']} \set*{ \sum_{w_{-i}} f^*(w_{-i}) \cdot \max(w_{-i})\vert_j \cdot \mathds{1}\paren*{ w_{-i} \notin E^{\und}_j(v_i) } } } \tag{\over{}} \\
&\hspace{0.5cm}+ \sum_{j = 1}^m \E_{v \sim \D^{n'}} \sq*{ \sum_{i = 1}^{n'} \sum_{w_{-i}} f^*(w_{-i}) \cdot \paren*{ v_{i,j} - \max(w_{-i})\vert_j } \cdot \mathds{1}\paren*{ w_{-i} \in E^{\nf}_j(v_i) } } \tag{\surplus{}} .
\end{align*}
It can now be shown that \over{} is at most $\srev(n')$. However, \surplus{} needs to be decomposed even more before it is analyzable. For this, we first use linearity of expectation to take the expectation over $v$ inside. As the summand corresponding to $i$ only depends on $v_i$, we get:
\[
\surplus = \sum_{j = 1}^m \sum_{i = 1}^{n'} \sum_{w_{-i}} \E_{v_i}\sq*{ f^*(w_{-i}) \cdot \paren*{ v_{i,j} - \max(w_{-i})\vert_j } \cdot \mathds{1}\paren*{ w_{-i} \in E^{\nf}_j(v_i) } } .
\]
Writing the expectation is a sum and noting that $w_{-i} \in E^{\nf}_j(v_i)$ only happens when $v_{i,j} \geq \max(w_{-i})\vert_j$ and $w_{-i} \in E^{\srp}_j(v_i)$ by \autoref{eq:events}, we get that:
\[
\surplus \leq \sum_{j = 1}^m \sum_{i = 1}^{n'} \sum_{w_{-i}} \sum_{ v_{i,j} \geq \max(w_{-i})\vert_j } f^*(w_{-i}) f_j(v_{i,j}) \cdot \paren*{ v_{i,j} - \max(w_{-i})\vert_j } \cdot \Pr_{v_{i,-j}}\paren*{ w_{-i} \in E^{\srp}_j(v_i) } .
\]

To continue, we define, for all $j \in [m]$, the function $r^*_{\ron, j}(x) = \max_{y > x} y \cdot \Pr_{y' \sim \D_j}\paren*{ y' \geq y }$. This definition is identical to that in \autoref{lemma:standard3} and is closely connected to the payment of the highest bidder in Ronen's auction for item $j$ when the second highest bid is $x$ \cite{Ronen01}.  We also define, for all $i, w_{-i}$, the quantity $r^{(i)}_{\ron}(w_{-i}) = \sum_{j = 1}^m r^*_{\ron, j}(\max(w_{-i})\vert_j)$ and, for all $j \in [m]$, the quantity $\T_{i,j}(w_{-i}) = r^{(i)}_{\ron}(w_{-i}) + \max(w_{-i})\vert_j$. Using $v_{i,-j}$ to denote the tuple $(v_{i,1}, \cdots, v_{i,j-1}, v_{i,j+1}, \cdots v_{i,m})$, we continue decomposing \surplus{} as:
\begin{align*}
\surplus &\leq \sum_{j = 1}^m \sum_{i = 1}^{n'} \sum_{w_{-i}} \sum_{ v_{i,j} > \T_{i,j}(w_{-i}) } f^*(w_{-i}) f_j(v_{i,j}) \cdot \paren*{ v_{i,j} - \max(w_{-i})\vert_j } \cdot \Pr_{v_{i,-j}}\paren*{ w_{-i} \in E^{\srp}_j(v_i) } \tag{\tail{}} \\
&\hspace{0.5cm} + \sum_{j = 1}^m \sum_{i = 1}^{n'} \sum_{w_{-i}} \sum_{ \max(w_{-i})\vert_j \leq v_{i,j} \leq \T_{i,j}(w_{-i}) } f^*(w_{-i}) f_j(v_{i,j}) \cdot \paren*{ v_{i,j} - \max(w_{-i})\vert_j } \tag{\core{}} .
\end{align*}
We call the first term above \tail{} and the second term as \core{}. We shall show that \tail{} is at most $\srev(n')$ while \core{} can be bounded as a function of $\bvcg(n')$ and $\srev(n')$. First, we state our final decomposition for $\ub(n')$:
\begin{equation}
\label{eq:ubsplit}
\ub(n') \leq \single + \under + \over + \tail + \core .
\end{equation}

To finish the proof of \autoref{lemma:step3step1}, we now show that each of the first four terms above is bounded by $\srev(n')$. 

\paragraph{Bounding \single{}.} If the term \single{} did not have the factor $\pregion_j(v_i)$ inside, it will just be maximum (over all auctions) value of (Myerson's) ironed virtual welfare, and we could use \autoref{prop:myerson} to finish the proof. As adding the factor $\pregion_j(v_i)$ can only decrease the value of \single{}, we derive:
\[
\single \leq \sum_{j = 1}^m \E_{v \sim \D^{n'}} \sq*{ \max_{i \in [n']} \set*{ \ivv_j(v_{i,j})^+ } } \leq \srev(n') \tag{\autoref{prop:myerson}} .
\]

\paragraph{Bounding \under{}.}  Roughly speaking, the term $v_{i,j}$ contributes to \under{} only if it is not the highest amongst $n'$ bids. As the fact that $v_{i,j}$ is not the highest amongst $n'$ bids implies that it is also not the highest amongst $n' + 1$ bids, we get:
\begin{align*}
\under &\leq \sum_{j = 1}^m \E_{v \sim \D^{n'}} \sq*{ \max_{i \in [n']} \set*{ \sum_{w_{-i}} f^*(w_{-i}) \cdot v_{i,j} \cdot \mathds{1}\paren*{ v_{i,j} < \max(w_{-i})\vert_j } } } \tag{\autoref{eq:events}} \\
&\leq \sum_{j = 1}^m \E_{v \sim \D^{n'}} \sq*{ \max_{i \in [n']} \set*{ \sum_w f^*(w) \cdot v_{i,j} \cdot \mathds{1}\paren*{ v_{i,j} \leq \max(w)\vert_j } } } .
\end{align*}
Now, consider each term $w$ inside the max as the bids of $n'$ bidders. In this interpretation (as formalized in \autoref{lemma:standard1}), the term inside the max is at most the revenue generated by a VCG auction where the reserve for item $j$ is $v_{i,j}$. Using \autoref{lemma:standard1}, we get:
\[
\under \leq \sum_{j = 1}^m \E_{v \sim \D^{n'}} \sq*{ \max_{i \in [n']} \set*{ \srev_j(n') } } = \sum_{j = 1}^m \srev_j(n') = \srev(n') .
\]

\paragraph{Bounding \over{}.} We first manipulate \over{} so that $w_{-i}$ can be moved outside the max. Using \autoref{eq:events}, we have:
\begin{align*}
\over &\leq \sum_{j = 1}^m \E_{v \sim \D^{n'}} \sq*{ \max_{i \in [n']} \set*{ \sum_w f^*(w) \cdot \max(w_{-i})\vert_j \cdot \mathds{1}\paren*{ v_{i,j} \geq \max(w_{-i})\vert_j } } } \\
&\leq \sum_w f^*(w) \cdot \sum_{j = 1}^m \E_{v \sim \D^{n'}} \sq*{ \max_{i \in [n']} \set*{ \max(w_{-i})\vert_j \cdot \mathds{1}\paren*{ v_{i,j} \geq \max(w_{-i})\vert_j } } } .
\end{align*}
We now analyze the term corresponding to each $w$ separately. For each $w$, consider a sequential posted price auction that sells each item separately. When selling item $j$, the auction visits the bidders in non-increasing order of $\max(w_{-i})\vert_j$ and offers them the item at price $\max(w_{-i})\vert_j$. The revenue generated by this auction is at least term corresponding to $w$ above. \autoref{lemma:standard2} formalizes this and gives:
\[
\over \leq \sum_w f^*(w) \cdot \srev(n') = \srev(n') .
\]
\paragraph{Bounding \tail{}.} At a high level, the term \tail{} is large only when bidder $i$ gets high utility from item $j$ but there exists an item $j' \neq j$ that gives even higher utility. This should be unlikely. More formally, by \autoref{eq:events} and a union bound, we have:
\[
\Pr_{v_{i,-j}}\paren*{ w_{-i} \in E^{\srp}_j(v_i) } \leq \sum_{j' \neq j} \Pr_{v_{i,j'}}\paren*{ v_{i,j} - \max(w_{-i})\vert_j \leq v_{i,j'} - \max(w_{-i})\vert_{j'} } .
\]
As \tail{} only sums over $v_{i,j} > \T_{i,j}(w_{-i}) \geq \max(w_{-i})\vert_j$, the definition of $r^*_{\ron, j'}(x)$ allows us to further bound this by: 
\begin{equation}
\label{eq:tailprob}
\begin{split}
\Pr_{v_{i,-j}}\paren*{ w_{-i} \in E^{\srp}_j(v_i) } &\leq \sum_{j' \neq j} \frac{ r^*_{\ron, j'}(\max(w_{-i})\vert_{j'}) }{ v_{i,j} - \max(w_{-i})\vert_j } \\
&\leq \frac{ r^{(i)}_{\ron}(w_{-i}) }{ v_{i,j} - \max(w_{-i})\vert_j } .
\end{split}
\end{equation}
Plugging \autoref{eq:tailprob} into the term \tail{}, we have:
\begin{equation}
\label{eq:tailsplit1}
\begin{split}
\tail &\leq \sum_{j = 1}^m \sum_{i = 1}^{n'} \sum_{w_{-i}} \sum_{ v_{i,j} > \T_{i,j}(w_{-i}) } f^*(w_{-i}) f_j(v_{i,j}) \cdot r^{(i)}_{\ron}(w_{-i}) \\
&\leq \sum_{j = 1}^m \sum_{i = 1}^{n'} \sum_{w_{-i}} f^*(w_{-i}) \cdot r^{(i)}_{\ron}(w_{-i}) \cdot \Pr_{v_{i,j}}\paren*{ v_{i,j} > \T_{i,j}(w_{-i}) } .
\end{split}
\end{equation}
Now, we claim that $r^{(i)}_{\ron}(w_{-i}) \cdot \Pr_{v_{i,j}}\paren*{ v_{i,j} > \T_{i,j}(w_{-i}) } \leq r^*_{\ron, j}(\max(w_{-i})\vert_j)$. In the case $\Pr_{v_{i,j}}\paren*{ v_{i,j} > \T_{i,j}(w_{-i}) } = 0$, this holds trivially. Otherwise, there exists $x \in \V_j$ be the smallest such that $x > \T_{i,j}(w_{-i}) = r^{(i)}_{\ron}(w_{-i}) + \max(w_{-i})\vert_j$ and we get:
\[
r^{(i)}_{\ron}(w_{-i}) \cdot \Pr_{v_{i,j}}\paren*{ v_{i,j} > \T_{i,j}(w_{-i}) } \leq x \cdot \Pr_{v_{i,j}}\paren*{ v_{i,j} \geq x } \leq r^*_{\ron, j}(\max(w_{-i})\vert_j) .
\]
We continue \autoref{eq:tailsplit1} as:
\[
\tail \leq \sum_{j = 1}^m \sum_{i = 1}^{n'} \sum_{w_{-i}} f^*(w_{-i}) \cdot r^*_{\ron, j}(\max(w_{-i})\vert_j) .
\]
The last expression is closely related to the revenue of a Ronen's auction \cite{Ronen01} that sells the items separately, and is captured in \autoref{lemma:standard3}. Using \autoref{lemma:standard3}, we conclude:
\[
\tail \leq \srev(n') .
\]
This concludes the proof of \autoref{lemma:step3step1}.
\end{proof}

\subsection{Step $2$ -- Bounding {\core{}}}
\label{sec:step3:core}

The next (and final) step in the proof of \autoref{lemma:step3} and \autoref{lemma:step3alt} is to upper bound the term \core{} that was left unanalyzed in \autoref{lemma:step3step1}. To this end, we first recall some definitions made in \autoref{sec:step3:decompose}. Recall that, for all $j \in [m]$, $r^*_{\ron, j}(x) = \max_{y > x} y \cdot \Pr_{y' \sim \D_j}\paren*{ y' \geq y }$ roughly (but not exactly) corresponds to the payment of the highest bidder in a Ronen's auction when the second highest bid is $x$. We also defined, for all $i, w_{-i}$ $r^{(i)}_{\ron}(w_{-i}) = \sum_{j = 1}^m r^*_{\ron, j}(\max(w_{-i})\vert_j)$ and for all $j \in [m]$, $\T_{i,j}(w_{-i}) = r^{(i)}_{\ron}(w_{-i}) + \max(w_{-i})\vert_j$. The term \core{} equals:
\[
\core = \sum_{j = 1}^m \sum_{i = 1}^{n'} \sum_{w_{-i}} \sum_{ \max(w_{-i})\vert_j \leq v_{i,j} \leq \T_{i,j}(w_{-i}) } f^*(w_{-i}) f_j(v_{i,j}) \cdot \paren*{ v_{i,j} - \max(w_{-i})\vert_j } .
\]
Observe that the term $\paren*{ v_{i,j} - \max(w_{-i})\vert_j }$ in the above equation is closely related to the utility that bidder with valuation $v_i$ gets from item $j$ in a VCG auction when the bids of the other bidders are $w_{-i}$. To capture this, we define the notation:
\begin{align*}
\util_{i,j,w_{-i}}(v_{i,j}) &= \max\paren*{ v_{i,j} - \max(w_{-i})\vert_j, 0 } \\
\caputil_{i,j,w_{-i}}(v_{i,j}) &= \util_{i,j,w_{-i}}(v_{i,j}) \cdot \mathds{1}\paren*{ \util_{i,j,w_{-i}}(v_{i,j}) \leq r^{(i)}_{\ron}(w_{-i}) } .
\end{align*}
These will primarily be used in the following form:
\begin{equation}
\label{eq:util}
\U_{i,w_{-i}}(v_i) = \sum_{j = 1}^m \util_{i,j,w_{-i}}(v_{i,j}) \hspace{1cm}\text{and}\hspace{1cm} \capU_{i,w_{-i}}(v_i) = \sum_{j = 1}^m \caputil_{i,j,w_{-i}}(v_{i,j}) .
\end{equation}
Using this notation, \core{} satsifies:
\begin{equation}
\label{eq:core}
\core = \sum_{i = 1}^{n'} \sum_{w_{-i}} f^*(w_{-i}) \cdot \E_{v_i}\sq*{ \capU_{i,w_{-i}}(v_i) } .
\end{equation}
Observe that, written this way, \core{} is closely related to the random variable $\capU_{i,w_{-i}}(v_i)$. It is in this form that we upper bound \core{} in \autoref{sec:step3step2} and \autoref{sec:step3step2alt}. But first, let us show using \autoref{lemma:variance-ub} that the variance of $\capU_{i,w_{-i}}(v_i)$ is small.

\begin{lemma}
\label{lemma:variancecaputil} 
It holds for all $i \in [n']$ and all $w_{-i}$ that:
\[
\Var_{v_i \sim \D}\paren*{ \capU_{i,w_{-i}}(v_i) } \leq 2 \cdot \paren*{ r^{(i)}_{\ron}(w_{-i}) }^2 .
\]
\end{lemma}
\begin{proof}
Recall that $\D = \bigtimes_{j = 1}^m \D_j$ is such that all the items are independent. Using the fact that variance is linear when over independent random variables (\autoref{fact:variance}, \autoref{item:varianceind}) and \autoref{eq:util}, we get:
\begin{equation}
\label{eq:variancecaputil1}
\Var_{v_i \sim \D}\paren*{ \capU_{i,w_{-i}}(v_i) } = \sum_{j = 1}^m \Var_{v_{i,j} \sim \D_j}\paren*{ \caputil_{i,j,w_{-i}}(v_{i,j}) } .
\end{equation}
Our goal now is to bound each term using \autoref{lemma:variance-ub}. To this end, note that $\caputil_{i,j,w_{-i}}(v_{i,j})$ is always at most $r^{(i)}_{\ron}(w_{-i})$ and thus, we can conclude that $\max_{v_{i,j}} \caputil_{i,j,w_{-i}}(v_{i,j}) \leq r^{(i)}_{\ron}(w_{-i})$. Moreover, we have for all $v_{i,j}$ that
\begin{multline*}
\caputil_{i,j,w_{-i}}(v_{i,j}) \cdot \Pr_{v'_{i,j} \sim \D_j}\paren*{ \caputil_{i,j,w_{-i}}(v'_{i,j}) \geq \caputil_{i,j,w_{-i}}(v_{i,j}) } \\
\leq \caputil_{i,j,w_{-i}}(v_{i,j}) \cdot \Pr_{v'_{i,j} \sim \D_j}\paren*{ v'_{i,j} \geq \max(w_{-i})\vert_j + \caputil_{i,j,w_{-i}}(v_{i,j}) } .
\end{multline*}
Now, if $\caputil_{i,j,w_{-i}}(v_{i,j}) = 0$, then, the right hand side is $0$ and consequently, is at most $r^*_{\ron, j}(\max(w_{-i})\vert_j)$. We show that the latter holds even when $\caputil_{i,j,w_{-i}}(v_{i,j}) > 0$. Indeed, we have:
\begin{align*}
&\caputil_{i,j,w_{-i}}(v_{i,j}) \cdot \Pr_{v'_{i,j} \sim \D_j}\paren*{ v'_{i,j} \geq \max(w_{-i})\vert_j + \caputil_{i,j,w_{-i}}(v_{i,j}) } \\
&\hspace{0.5cm} \leq \paren*{ \caputil_{i,j,w_{-i}}(v_{i,j}) + \max(w_{-i})\vert_j } \cdot \Pr_{v'_{i,j} \sim \D_j}\paren*{ v'_{i,j} \geq \max(w_{-i})\vert_j + \caputil_{i,j,w_{-i}}(v_{i,j}) } \\
&\hspace{0.5cm} \leq r^*_{\ron, j}(\max(w_{-i})\vert_j) \tag{Definition of $r^*_{\ron, j}(\cdot)$} .
\end{align*}
Thus, we can conclude that:
\[
\max_{v_{i,j}} \caputil_{i,j,w_{-i}}(v_{i,j}) \cdot \Pr_{v'_{i,j} \sim \D_j}\paren*{ \caputil_{i,j,w_{-i}}(v'_{i,j}) \geq \caputil_{i,j,w_{-i}}(v_{i,j}) } \leq r^*_{\ron, j}(\max(w_{-i})\vert_j).
\]
Plugging this and $\max_{v_{i,j}} \caputil_{i,j,w_{-i}}(v_{i,j}) \leq r^{(i)}_{\ron}(w_{-i})$ into \autoref{lemma:variance-ub}, we get:
\[
\Var_{v_{i,j} \sim \D_j}\paren*{ \caputil_{i,j,w_{-i}}(v_{i,j}) } \leq 2 \cdot r^*_{\ron, j}(\max(w_{-i})\vert_j) \cdot r^{(i)}_{\ron}(w_{-i}) .
\]
Plugging into \autoref{eq:variancecaputil1}, we get:
\[
\Var_{v_i \sim \D}\paren*{ \capU_{i,w_{-i}}(v_i) } \leq 2 \cdot r^{(i)}_{\ron}(w_{-i}) \cdot \sum_{j = 1}^m r^*_{\ron, j}(\max(w_{-i})\vert_j) \leq 2 \cdot \paren*{ r^{(i)}_{\ron}(w_{-i}) }^2 .
\]

\end{proof}

\subsubsection{Bounding {\core{}} for \autoref{lemma:step3}}
\label{sec:step3step2}

In this section, we finish our proof of \autoref{lemma:step3} by upper bounding the right hand side of \autoref{eq:core} by the revenue of a BVCG auction (and $\srev(n')$). Specifically, we shall consider a BVCG auction with $n'$ bidders, where the fee charged for player $i$, when the types of the other bidders are $w_{-i}$ is:
\[
\fee_{i,w_{-i}} = \max\paren*{ \E_{v_i}\sq*{ \capU_{i,w_{-i}}(v_i) } - 2 \cdot r^{(i)}_{\ron}(w_{-i}), 0 } .
\]
The following lemma shows that most bidders will agree to pay this extra fee, and thus, expectation of the total fee is at most $2 \cdot \bvcg(n')$.
\begin{lemma}
\label{lemma:bvcgcore}
It holds that:
\[
\sum_{i = 1}^{n'} \sum_{w_{-i}} f^*(w_{-i}) \cdot \fee_{i,w_{-i}} \leq 2 \cdot \bvcg(n') .
\]
\end{lemma}
\begin{proof}
Consider the BVCG auction defined by $\fee_{i,w_{-i}}$. That is, consider the auction where the auctioneer first asks all bidders $i \in [n']$ for their bids $w_i$ and runs a VCG auction based on these bids. If bidder $i \in [n']$ is not allocated any items in the VCG auction, he departs without paying anything. Otherwise, he gets all the items allocated to him in the VCG auction if and only if he agrees to pay an amount equal to $\fee_{i,w_{-i}}$ in addition to the prices charged by the VCG auction. 

This auction is truthful as we ensure that $\fee_{i,w_{-i}} \geq 0$. Moreover, if bidder $i$ does not pay at least $\fee_{i,w_{-i}}$, we must have that his utility from the VCG auction is (strictly) smaller that $\fee_{i,w_{-i}}$. Thus, we get the following lower bound on $\bvcg(n')$.
\begin{align*}
\bvcg(n') &\geq \sum_{i = 1}^{n'} \sum_{w \in \V^{n'}} f^*(w) \cdot \fee_{i,w_{-i}} \cdot \mathds{1}\paren*{ \fee_{i,w_{-i}} \leq \U_{i,w_{-i}}(w_i) } \\
&\geq \sum_{i = 1}^{n'} \sum_{w_{-i}} f^*(w_{-i}) \cdot \fee_{i,w_{-i}} \cdot \Pr_{w_i}\paren*{ \fee_{i,w_{-i}} \leq \U_{i,w_{-i}}(w_i) } \\
&\geq \sum_{i = 1}^{n'} \sum_{w_{-i}} f^*(w_{-i}) \cdot \fee_{i,w_{-i}} \cdot \Pr_{w_i}\paren*{ \fee_{i,w_{-i}} \leq \capU_{i,w_{-i}}(w_i) } ,
\end{align*}
where the last step is because $\U$ upper bounds $\capU$. The next step is to lower bound the probability on the right hand side. We do this using Chebyshev's inequality (\autoref{fact:variance}, \autoref{item:chebyshev}) and use the variance bound in \autoref{lemma:variancecaputil}. We have:
\[
\Pr_{w_i}\paren*{ \capU_{i,w_{-i}}(w_i) < \fee_{i,w_{-i}} } \leq \frac{1}{2} .
\]
Plugging in, we have:
\[
\bvcg(n') \geq \sum_{i = 1}^{n'} \sum_{w_{-i}} f^*(w_{-i}) \cdot \fee_{i,w_{-i}} \cdot \frac{1}{2} .
\]
and the lemma follows.
\end{proof}

We now present our proof of \autoref{lemma:step3}.
\begin{proof}[Proof of \autoref{lemma:step3}]
From \autoref{eq:core} and the definition of $\fee_{i,w_{-i}}$, we have:
\begin{align*}
\core &\leq \sum_{i = 1}^{n'} \sum_{w_{-i}} f^*(w_{-i}) \cdot \paren*{ \fee_{i,w_{-i}} + 2 \cdot r^{(i)}_{\ron}(w_{-i}) } \\
&\leq \sum_{i = 1}^{n'} \sum_{w_{-i}} f^*(w_{-i}) \cdot \fee_{i,w_{-i}} + 2 \cdot \sum_{i = 1}^{n'} \sum_{w_{-i}} f^*(w_{-i}) \cdot r^{(i)}_{\ron}(w_{-i}) \\
&\leq \sum_{i = 1}^{n'} \sum_{w_{-i}} f^*(w_{-i}) \cdot \fee_{i,w_{-i}} + 2 \cdot \sum_{j = 1}^m \sum_{i = 1}^{n'} \sum_{w_{-i}} f^*(w_{-i}) \cdot r^*_{\ron, j}(\max(w_{-i})\vert_j) .
\end{align*}
These two terms can be bounded by \autoref{lemma:bvcgcore} and \autoref{lemma:standard3} respectively yielding $\core \leq 2 \cdot \bvcg(n') + 2 \cdot \srev(n')$. Plugging into \autoref{lemma:step3step1}, we get:
\[
\ub(n'', n') \leq 4 \cdot \srev(n') + \core \leq 2 \cdot \bvcg(n') + 6 \cdot \srev(n') .
\]

\end{proof}

\subsubsection{Bounding {\core{}} for \autoref{lemma:step3alt}}
\label{sec:step3step2alt}

Now, we finish our proof of \autoref{lemma:step3alt} by upper bounding the right hand side of \autoref{eq:core} by the revenue of a {\em prior-independent} BVCG auction. The auction defined in \autoref{sec:step3step2} was not prior independent as to compute the fees charged to the bidders required knowledge of the distribution $\D$. Our main idea follows \cite{GoldnerK16}, we construct an auction with $n' + 1$ bidders, and treat the last bidder as `special'. This special bidder does not receive any items or pay anything, but his bids allow us to get a good enough estimate of the distribution $\D$.  

We shall reserve $s$ to denote the bid of the special bidder and $w \in \V^{n'}$ will denote the bids of the other bidders. For $i \in [n']$, the notation $w_i$ will (as before) denote the bid of player $i$, while $w_{-i}$ will denote the bids of all the other players {\em excluding the special player}. This time the fee for player $i \in [n']$ is defined as (recall \autoref{eq:util}):
\begin{equation}
\label{eq:feepi}
\fee_{i,w_{-i},s} = \U_{i,w_{-i}}(s) .
\end{equation}
Importantly, this is determined by the bids of the bidders and is independent of $\D$. We also define, for all $i$ and $w_{-i}$, the set $\N_{i, w_{-i}}$ as follows:
\begin{equation}
\label{eq:nice}
\N_{i, w_{-i}} = \set*{ v \in \V~\middle\vert~\capU_{i,w_{-i}}(v) \geq \frac{1}{2} \cdot \E_{v' \sim \D}\sq*{ \capU_{i,w_{-i}}(v') } } .
\end{equation}

We now show a prior-independent analogue of \autoref{lemma:bvcgcore}.

\begin{lemma}
\label{lemma:bvcgcorepi}
It holds that:
\[
\sum_{i = 1}^{n'} \sum_{w_{-i}} f^*(w_{-i}) \cdot \Pr_{v' \sim \D}\paren*{ v' \in \N_{i, w_{-i}} }^2 \cdot \E_{v' \sim \D}\sq*{ \capU_{i,w_{-i}}(v') } \leq 4 \cdot \pibvcg(n' + 1) .
\]
\end{lemma}
\begin{proof}
We follow the proof approach in \autoref{lemma:bvcgcore} but this time use the fees defined in \autoref{eq:feepi} as they are prior-independent. More specifically, we consider the auction that first receives the bids $w$ and $s$ of the non-special and special players respectively and runs a VCG auction based on $w$. Thus, the special bidder never receives or pays anything. If bidder $i \in [n']$ is not allocated any items in the VCG auction, he departs without paying anything. Otherwise, he gets all the items allocated to him in the VCG auction if and only if he agrees to pay an amount equal to $\fee_{i,w_{-i},s}$ in addition to the prices charged by the VCG auction. 

This auction is truthful as we ensure that $\fee_{i,w_{-i},s} \geq 0$. Moreover, if bidder $i$ does not pay at least the amount $\fee_{i,w_{-i},s}$, we must have that his utility from the VCG auction is (strictly) smaller than $\fee_{i,w_{-i},s}$. From \autoref{eq:feepi}, we get the following lower bound on the revenue of this auction:
\begin{align*}
\pibvcg(n' + 1) &\geq \sum_{i = 1}^{n'} \sum_{s \in \V} \sum_{w \in \V^{n'}} f(s) f^*(w) \cdot \U_{i,w_{-i}}(s) \cdot \mathds{1}\paren*{ \U_{i,w_{-i}}(s) \leq \U_{i,w_{-i}}(w_i) } \\
&\geq \sum_{i = 1}^{n'} \sum_{w_{-i}} \sum_{s, w_i \in \N_{i, w_{-i}}} f(s) f(w_i) f^*(w_{-i}) \cdot \U_{i,w_{-i}}(s) \cdot \mathds{1}\paren*{ \U_{i,w_{-i}}(s) \leq \U_{i,w_{-i}}(w_i) } .
\end{align*}
As $\U$ upper bounds $\capU$ and we only consider $s \in \N_{i, w_{-i}}$, we have $\U_{i,w_{-i}}(s) \geq \capU_{i,w_{-i}}(s) \geq \frac{1}{2} \cdot \E_{v' \sim \D}\sq*{ \capU_{i,w_{-i}}(v') }$. Plugging in, we have:
\begin{multline*}
\pibvcg(n' + 1) \\
\geq \frac{1}{2} \cdot \sum_{i = 1}^{n'} \sum_{w_{-i}} f^*(w_{-i}) \cdot \E_{v' \sim \D}\sq*{ \capU_{i,w_{-i}}(v') } \cdot \sum_{s, w_i \in \N_{i, w_{-i}}} f(s) f(w_i) \cdot \mathds{1}\paren*{ \U_{i,w_{-i}}(s) \leq \U_{i,w_{-i}}(w_i) } .
\end{multline*}
By symmetry, we conclude that:
\[
\pibvcg(n' + 1) \geq \frac{1}{4} \cdot \sum_{i = 1}^{n'} \sum_{w_{-i}} f^*(w_{-i}) \cdot \Pr_{v' \sim \D}\paren*{ v' \in \N_{i, w_{-i}} }^2 \cdot \E_{v' \sim \D}\sq*{ \capU_{i,w_{-i}}(v') } .
\]
\end{proof}

We now present our proof of \autoref{lemma:step3alt}.
\begin{proof}[Proof of \autoref{lemma:step3alt}]
Call a pair $i, w_{-i}$ ``high'' if
\begin{equation}
\label{eq:high}
\E_{v' \sim \D}\sq*{ \capU_{i,w_{-i}}(v') } \geq 6 \cdot r^{(i)}_{\ron}(w_{-i}) ,
\end{equation}
and call it ``low'' otherwise. Using Chebyshev's inequality (\autoref{fact:variance}, \autoref{item:chebyshev}) and the variance bound in \autoref{lemma:variancecaputil}, we have for all high $(i, w_{-i})$ that:
\[
1 - \Pr_{v' \sim \D}\paren*{ v' \in \N_{i, w_{-i}} } \leq \frac{ 4 \cdot \Var_{v_i \sim \D}\paren*{ \capU_{i,w_{-i}}(v_i) } }{ \paren*{ \E_{v' \sim \D}\sq*{ \capU_{i,w_{-i}}(v') }^2 } } \leq \frac{2}{9} .
\]
Thus, if the pair $(i, w_{-i})$ is high, we get that $\Pr_{v' \sim \D}\paren*{ v' \in \N_{i, w_{-i}} }$ is at least $\frac{7}{9}$. We now bound \core{} from \autoref{eq:core} and finish the proof. We have:
\begin{align*}
\core &\leq \sum_{\text{high~} (i, w_{-i})} f^*(w_{-i}) \cdot \E_{v' \sim \D}\sq*{ \capU_{i,w_{-i}}(v') } + \sum_{\text{low~} (i, w_{-i})} f^*(w_{-i}) \cdot \E_{v' \sim \D}\sq*{ \capU_{i,w_{-i}}(v') } \\
&\leq \frac{81}{49} \cdot \sum_{\text{high~} (i, w_{-i})} f^*(w_{-i}) \cdot \Pr_{v' \sim \D}\paren*{ v' \in \N_{i, w_{-i}} }^2 \cdot \E_{v' \sim \D}\sq*{ \capU_{i,w_{-i}}(v') } \\
&\hspace{1cm} + 6 \cdot \sum_{\text{low~} (i, w_{-i})} f^*(w_{-i}) \cdot r^{(i)}_{\ron}(w_{-i}) ,
\end{align*}
where, for high $(i, w_{-i})$, we plug in $\Pr_{v' \sim \D}\paren*{ v' \in \N_{i, w_{-i}} } \geq \frac{7}{9}$, while for low $(i, w_{-i})$, we use \autoref{eq:high}. This gives:
\[
\core \leq \frac{81}{49} \cdot \sum_{i = 1}^{n'} \sum_{w_{-i}} f^*(w_{-i}) \cdot \Pr_{v' \sim \D}\paren*{ v' \in \N_{i, w_{-i}} }^2 \cdot \E_{v' \sim \D}\sq*{ \capU_{i,w_{-i}}(v') } + 6 \cdot \sum_{i = 1}^{n'} \sum_{w_{-i}} f^*(w_{-i}) \cdot r^{(i)}_{\ron}(w_{-i}) .
\]
Using \autoref{lemma:bvcgcorepi} and using the definition of $r^{(i)}_{\ron}(w_{-i})$, we get: 
\[
\core \leq 7 \cdot \pibvcg(n' + 1) + 6 \cdot \sum_{j = 1}^m \sum_{i = 1}^{n'} \sum_{w_{-i}} f^*(w_{-i}) \cdot r^*_{\ron, j}(\max(w_{-i})\vert_j) .
\]
Using \autoref{lemma:standard3} on the second term, we have $\core \leq 7 \cdot \pibvcg(n' + 1) + 6 \cdot \srev(n')$. Plugging into \autoref{lemma:step3step1}, we get $\ub(n'', n') \leq 4 \cdot \srev(n') + \core \leq 7 \cdot \pibvcg(n' + 1) + 10 \cdot \srev(n')$. As we assumed that all the items are regular, we have from \autoref{prop:bk} that $\srev(n') \leq \vcg(n' + 1) \leq \pibvcg(n' + 1)$. This yields:
\[
\ub(n'', n') \leq 17 \cdot \pibvcg(n' + 1) .
\]
\end{proof}

\bibliographystyle{alpha}
\bibliography{MasterBib}

\appendix


\section{Some Lower Bounds on $\srev(\cdot)$} 
\label{app:standard}

In this section, we analyze the revenue of some auctions that sell the items separately. By definition, the revenue of any such auction is a lower bound for $\srev(\cdot)$. All lemmas in this section are for a fixed auction setting $(n, m, \D)$ (see \autoref{sec:prelim}).

\begin{lemma}[VCG with reserves]
\label{lemma:standard1}
Fix item $j \in [m]$. For all $x \geq 0$, it holds that:
\[
\sum_{v \in \V^n} f^*(v) \cdot x \cdot \mathds{1}\paren*{ \max(v)\vert_j \geq x } \leq \srev_j(n) .
\]
\end{lemma}
\begin{proof}
Consider the auction that sells item $j$ through a VCG auction with reserve $x$. Namely, it solicits bits $v_{i,j}$ for item $j$ for each bidder $i \in [m]$ and proceeds as follows: If the highest bid is at least $x$, then allocate this item to the highest bidder for a price equal to the maximum of $x$ and the second highest bid. Otherwise, the item stays unallocated. Clearly, the auction is truthful and generates revenue at least:
\[
\sum_{v \in \V^n} f^*(v) \cdot x \cdot \mathds{1}\paren*{ \max(v)\vert_j \geq x } .
\]
Thus, we can upper bound the above quantity by $\srev_j(n)$ and the lemma follows.
\end{proof}

\begin{lemma}[Sequential Posted Price]
\label{lemma:standard2}
Let non-negative numbers $\{x_{i,j}\}_{i \in [n], j \in [m]}$ be given. It holds that:
\[
\sum_{j = 1}^m \sum_{v \in \V^n} f^*(v) \cdot \max_{i \in [n]} \set*{ x_{i,j} \cdot \mathds{1}\paren*{ v_{i,j} \geq x_{i,j} } } \leq \srev(n) .
\]
\end{lemma}
\begin{proof}
Consider the auction that sells each item $j \in [m]$ separately through the following auction: It goes over all the bidders in decreasing order of $x_{i,j}$, bidder $i$ can either take the item and pay price $x_{i,j}$, in which case the auction terminates, or skip the item, in which case the auction goes to the next bidder. Clearly, the auction is truthful and generates revenue at least:
\[
\sum_{j = 1}^m \sum_{v \in \V^n} f^*(v) \cdot \max_{i \in [n]} \set*{ x_{i,j} \cdot \mathds{1}\paren*{ v_{i,j} \geq x_{i,j} } } .
\]
Thus, we can upper bound the above quantity by $\srev(n)$ and the lemma follows.
\end{proof}

\begin{lemma}[Ronen's auction \cite{Ronen01}]
\label{lemma:standard3}
For all $j \in [m]$ and $x \geq 0$, define $r^*_{\ron, j}(x) = \max_{y > x} y \cdot \Pr_{y' \sim \D_j}\paren*{ y' \geq y }$. It holds that:
\[
\sum_{j = 1}^m \sum_{i = 1}^n \sum_{v_{-i} \in \V^{n-1}} f^*(v_{-i}) \cdot r^*_{\ron, j}(\max(v_{-i})\vert_j) \leq \srev(n) .
\]
\end{lemma}
\begin{proof}
Consider the auction that sells each item $j \in [m]$ separately through the following auction: First, it solicits bids $v_{i,j}$ for item $j$ from each bidder $i \in [n]$. Then, for $i \in [n]$, it sets $y^*_{i,j}(v_{-i})$ to be\footnote{We write $y^*_{i,j}$ as a function of $v_{-i}$ but note that it only depends on the bidders' bids for item $j$.} the maximizer in the definition of $r^*_{\ron, j}(\max(v_{-i})\vert_j)$, and offers each bidder $i$ to purchase item $j$ at a price of $y^*_{i,j}(v_{-i})$. As $y^*_{i,j}(v_{-i}) > \max(v_{-i})\vert_j$ by definition, at most one bidder will ever purchase the item and the auction is well defined (\autoref{eq:allocconstraint}). 

Also, as the price offered to bidder $i$ does not depend on his bid, the auction is also truthful. Thus, its revenue is a lower bound on $\srev(n)$ and we get:
\begin{align*}
\srev(n) &\geq \sum_{j = 1}^m \sum_{v \in \V^n} f^*(v) \cdot \sum_{i = 1}^n y^*_{i,j}(v_{-i}) \cdot \mathds{1}\paren*{ v_{i,j} \geq y^*_{i,j}(v_{-i}) } \\
&\geq \sum_{j = 1}^m \sum_{i = 1}^n \sum_{v_{-i} \in \V^{n-1}} f^*(v_{-i}) \cdot y^*_{i,j}(v_{-i}) \cdot \Pr_{v_i \in \V}\paren*{ v_{i,j} \geq y^*_{i,j}(v_{-i}) } \\
&\geq \sum_{j = 1}^m \sum_{i = 1}^n \sum_{v_{-i} \in \V^{n-1}} f^*(v_{-i}) \cdot r^*_{\ron, j}(\max(v_{-i})\vert_j) .
\end{align*}

\end{proof}


\section{Proof of Corollary 27 of \cite{CaiDW16}} 
\label{app:CDW} 

This section recalls the proof of Corollary $27$ from \cite{CaiDW16} as \autoref{lemma:cor27}. Our presentation is different from \cite{CaiDW16} as we do not need their ideas in full generality.

\begin{lemma}
\label{lemma:cor27}
Let $(n, m, \D)$ be an auction setting as in \autoref{sec:prelimauc}. Let $n' > 0$ and suppose that for all $i \in [n]$, valuations $w_{-i} \in \V^{n' - 1}$ are given. For all $(\barpi, \barp)$ that correspond to a truthful auction $\auc$, we have that:
\[
\rev(\auc, n) \leq \sum_{i = 1}^n \sum_{j = 1}^m \E_{v_i} \sq*{ \barpi_{i,j}(v_i) \cdot \paren*{ v_{i,j} \cdot \mathds{1}\paren*{ v_i \notin \region^{(n')}_j(w_{-i}) } + \ivv_j(v_{i,j})^+ \cdot \mathds{1}\paren*{ v_i \in \region^{(n')}_j(w_{-i}) } } } .
\]
\end{lemma}
\begin{proof} 
We start with some notation. We use $v_{\varnothing}$ to denote a dummy valuation for the bidders and adopt the convention $\barpi_{i,j}(v_{\varnothing}) =  \barp_i(v_{\varnothing}) = 0$ for all $i \in [n]$, $j \in [m]$. Suppose that non-negative numbers $\Lambda = \set*{ \lambda_i(v_i, v'_i) }_{i \in [n], v_i \in \V, v'_i \in \V \cup \set*{ v_{\varnothing} }}$ are given that satisfy for all $i \in [n]$ and $v_i \in \V$ that:
\begin{equation}
\label{eq:flowconserve}
f(v_i) - \sum_{v'_i \in \V \cup \set*{ v_{\varnothing} } } \lambda_i(v_i, v'_i) + \sum_{v'_i \in \V} \lambda_i(v'_i, v_i) = 0 .
\end{equation}
As $(\barpi, \barp)$ correspond to a truthful auction $\auc$, we have from \autoref{eq:revenue} that $\rev(\auc, n) = \sum_{i = 1}^n \E_{v_i \sim \D} \sq*{ \barp_i(v_i) }$. Continuing using the non-negativity of $\lambda_i(v_i, v'_i)$ and \autoref{eq:truthful}, we have:
\begin{multline*}
\rev(\auc, n) \leq \sum_{i = 1}^n \sum_{v_i \in \V} f(v_i) \cdot \barp_i(v_i) \\
+ \sum_{i = 1}^n \sum_{v_i \in \V} \sum_{v'_i \in \V \cup \set*{ v_{\varnothing} } } \lambda_i(v_i, v'_i) \cdot \paren*{ \sum_{j = 1}^m \paren*{ \barpi_{i,j}(v_i) - \barpi_{i,j}(v'_i) } \cdot v_{i,j} - \paren*{ \barp_i(v_i) - \barp_i(v'_i) } } .
\end{multline*}
This can be rearranged to:
\begin{multline*}
\rev(\auc, n) \leq \sum_{i = 1}^n \sum_{v_i \in \V} \paren*{ f(v_i) - \sum_{v'_i \in \V \cup \set*{ v_{\varnothing} } } \lambda_i(v_i, v'_i) + \sum_{v'_i \in \V} \lambda_i(v'_i, v_i) } \cdot \barp_i(v_i) \\
+ \sum_{i = 1}^n \sum_{v_i \in \V} \sum_{j = 1}^m \paren*{ \sum_{v'_i \in \V \cup \set*{ v_{\varnothing} } } \lambda_i(v_i, v'_i) \cdot v_{i,j} - \sum_{v'_i \in \V} \lambda_i(v'_i, v_i) \cdot v'_{i,j} } \cdot \barpi_{i,j}(v_i) .
\end{multline*}
Plugging in \autoref{eq:flowconserve}, we get:
\[
\rev(\auc, n) \leq \sum_{i = 1}^n \sum_{v_i \in \V} \sum_{j = 1}^m \paren*{ f(v_i) \cdot v_{i,j} - \sum_{v'_i \in \V} \lambda_i(v'_i, v_i) \cdot \paren*{ v'_{i,j} - v_{i,j} } } \cdot \barpi_{i,j}(v_i) .
\]
Rearranging again, and denoting by $\fv^{\Lambda}_{i,j}(v_i) = v_{i,j} - \frac{1}{ f(v_i) } \cdot \sum_{v'_i \in \V} \lambda_i(v'_i, v_i) \cdot \paren*{ v'_{i,j} - v_{i,j} }$, we get:
\begin{equation}
\label{eq:flowvv}
\rev(\auc, n) \leq \sum_{i = 1}^n \sum_{v_i \in \V} \sum_{j = 1}^m f(v_i) \cdot \barpi_{i,j}(v_i) \cdot \fv^{\Lambda}_{i,j}(v_i) .
\end{equation}

Observe that \autoref{eq:flowvv} holds for any $\Lambda$ that is non-negative and satisfies \autoref{eq:flowconserve}. In order to show \autoref{lemma:cor27}, we construct a suitable $\Lambda$ and apply \autoref{eq:flowvv}. This is done by defining $\Lambda'$ and $\Lambda^*$ as below and setting $\Lambda = \Lambda' + \Lambda^*$.

\paragraph{Defining $\Lambda'$.} We start with some notation. For $j \in [m]$, let $f_{-j}(\cdot)$ denote the probability mass function of the distribution $\bigtimes_{j' \neq j} \D_{j'}$. Also, for $j \in [m]$ and $v_i \in \V$, let $\dec_j(v_i)$ be defined to be $v_{\varnothing}$ if $v_{i,j} = \min \V_j$. Otherwise define $\dec_{j,k}(v_i) = v_{i,k}$ for all $k \neq j$ and $\dec_{j,j}(v_i) = \max_{x \in \V_j, x < v_{i,j}} x$. Recall the definition of the regions $\set*{ \region^{(n')}_j(w_{-i}) }_{j \in \set*{0} \cup [m]}$ from \autoref{eq:region} and for all $i \in [n], v_i \in \V, v'_i \in \V \cup \set*{ v_{\varnothing} }$, define the numbers:
\[
\lambda'_i(v_i, v'_i) = \begin{cases}
f(v_i), &\text{~if~} v_i \in \region^{(n')}_0(w_{-i}) \text{~and~} v'_i = v_{\varnothing} \\
\Pr_{y \sim \D_j}\paren*{ y \geq v_{i,j} } \cdot f_{-j}(v_{i,-j}), &\text{~if~} \exists j \in [m] : v_i, v'_i \in \region^{(n')}_j(w_{-i}) \wedge v'_i = \dec_j(v_i) \\
\Pr_{y \sim \D_j}\paren*{ y \geq v_{i,j} } \cdot f_{-j}(v_{i,-j}), &\text{~if~} \exists j \in [m] : v_i \in \region^{(n')}_j(w_{-i}) \wedge v'_i = v_{\varnothing} \\
&\hspace{1cm} \text{~and~} \dec_j(v_i) \notin \region^{(n')}_j(w_{-i}) \\
0, &\text{~otherwise} 
\end{cases} .
\]

\paragraph{Defining $\Lambda^*$.} For all $i \in [n]$, we define $\lambda^*_i(\cdot)$ using the procedure described in \autoref{algo:lambda}. In \autoref{line:algocall} of \autoref{algo:lambda}, when we say we invoke \autoref{algo:ivv} restricted to values at least $x$, we mean that \autoref{line:setystar} of \autoref{algo:ivv} would only include values that are at least $x$ in the $\argmax$ and \autoref{line:abort} of \autoref{algo:ivv} will abort as soon as $y^* = x$ (instead of when $y^* = \min(\V_j)$). \autoref{algo:ivv} guarantees that the output $\vv_j^{v_{i,-j}}(\cdot)$ produced in this manner is a lower bound of $\ivv_j(\cdot)$, and therefore, also a lower bound of $\ivv_j(\cdot)^+$, for all values at least $x$. Moreover it satisfies, for all $y \geq x$ and with equality when $y=x$, that:
\begin{equation}
\label{eq:lambdavv}
\sum_{y' \geq y \in \V_j} f_j(y') \cdot \vv_j^{v_{i,-j}}(y') \geq \sum_{y' \geq y \in \V_j} f_j(y') \cdot \vv_j(y') .
\end{equation}

\begin{algorithm}
\caption{Computing $\lambda^*_i(\cdot)$ for $i \in [n]$.}
\label{algo:lambda}
\begin{algorithmic}[1]

\State Set $\lambda^*_i(v_i, v'_i) = 0$ for all $v_i \in \V$ and $v'_i \in \V \cup \set*{ v_{\varnothing} }$.

\For{$j \in [m]$}

\For{$v_{i,-j} \in \V_{-j}$}

\State $S \gets \min \set*{ x \in \V_j \mid (x,v_{i,-j}) \in \region^{(n')}_j(w_{-i}) }$. If $S = \emptyset$, continue to next iteration.

\State $x^* \gets \min(S)$.

\State $\vv_j^{v_{i,-j}}(\cdot) \gets$ the output of \autoref{algo:ivv} when restricted to values at least $x^*$. \label{line:algocall}

\For{$x \in \V_j$ such that  $x^* \leq x < \max \V_j$}

\State $x' \gets$ smallest element $> x$ in $\V_j$.

\State Set both $\lambda^*_i((x, v_{i,-j}), (x', v_{i,-j}))$ and $\lambda^*_i((x', v_{i,-j}), (x, v_{i,-j}))$ to
\[
\frac{ f_{-j}(v_{i,-j}) }{ x' - x } \cdot \sum_{x'' > x \in \V_j} f_j(x'') \paren*{ \vv_j^{v_{i,-j}}(x'') - \vv_j(x'') } .
\]

\EndFor

\EndFor

\EndFor

\end{algorithmic}
\end{algorithm}

We now finish the proof of \autoref{lemma:cor27}. Having defined $\Lambda'$ and $\Lambda^{\star}$, we first observe that they are both non-negative ($\Lambda^*$ is non-negative due to \autoref{eq:lambdavv}). Moreover, observe that setting $\Lambda = \Lambda' + \Lambda^*$ satisfies \autoref{eq:flowconserve}. Plugging into \autoref{eq:flowvv}, we get:
\[
\rev(\auc, n) \leq \sum_{i = 1}^n \sum_{j = 1}^m \E_{v_i}\sq*{ \barpi_{i,j}(v_i) \cdot \fv^{\Lambda}_{i,j}(v_i) } .
\]
Where, using \autoref{eq:lambdavv} and \autoref{def:vv}, the value $\fv^{\Lambda}_{i,j}(v_i)$ can be simplified to:
\[
\fv^{\Lambda}_{i,j}(v_i) = v_{i,j} \cdot \mathds{1}\paren*{ v_i \notin \region^{(n')}_j(w_{-i}) } + \vv_j^{v_{i,-j}}(v_{i,j}) \cdot \mathds{1}\paren*{ v_i \in \region^{(n')}_j(w_{-i}) } .
\]
Plugging in and using the fact that $\vv_j^{v_{i,-j}}(\cdot) \leq \ivv_j(\cdot)^+$, we get:
\[
\rev(\auc, n) \leq \sum_{i = 1}^n \sum_{j = 1}^m \E_{v_i} \sq*{ \barpi_{i,j}(v_i) \cdot \paren*{ v_{i,j} \cdot \mathds{1}\paren*{ v_i \notin \region^{(n')}_j(w_{-i}) } + \ivv_j(v_{i,j})^+ \cdot \mathds{1}\paren*{ v_i \in \region^{(n')}_j(w_{-i}) } } } .
\]

\end{proof}

\end{document}